\documentclass[prx, twocolumn, superscriptaddress, dvipsnames]{revtex4-2}
\pdfoutput=1
\usepackage[utf8]{inputenc}
\usepackage[english]{babel}
\usepackage{amsmath}
\usepackage{amssymb}
\usepackage{amsthm}
\usepackage{amsfonts}
\usepackage{graphicx}
\usepackage{caption}
\usepackage{bbm}
\usepackage{makecell}
\usepackage{subcaption}
\usepackage{xcolor}
\usepackage{bm}
\usepackage{here}
\usepackage{soul}
\usepackage{braket}
\usepackage{bbold}

\usepackage[colorlinks,allcolors=blue]{hyperref}
\usepackage[noabbrev, capitalize, nameinlink]{cleveref}

\newtheorem{thm}{Theorem}[section]

\newtheorem{prop}[thm]{Proposition}

\creflabelformat{equation}{#2\textup{#1}#3}
\theoremstyle{definition}
\newtheorem{defn}[thm]{Definition}

\usepackage{tikz}
\usetikzlibrary{decorations.pathreplacing,calligraphy,decorations.markings}

\definecolor{tensorcolor}{rgb}{0.65,0.77,0.95}
\definecolor{mpdocolor}{HTML}{FCDE70}
\definecolor{whampdocolor}{HTML}{FCDE70}
\definecolor{whampdocolorw}{HTML}{EEF7FF}
\definecolor{mpdotcolor}{rgb}{1,0.98,0.94}
\definecolor{btensorcolor}{rgb}{0.65,0.50,0.69}
\definecolor{whitetensorcolor}{HTML}{F8F8F8}
\definecolor{diamondcolor}{HTML}{E6F1ED}
\definecolor{unitarycolor}{rgb}{0.8,0.5,.5}
\definecolor{lcolor}{HTML}{D9EAFD}

\newcommand\doubledx{1.6}

\newcommand\whadx{0.988}
\newcommand\whady{0.812}
\newcommand\mthick{}

\newcommand{\op}[3]{
\begin{scope}[shift={(#1)}]
	\filldraw[\mthick, fill=white](0,0) circle[radius=#2];
	\draw (0,0) node {\scriptsize #3};
\end{scope}
}

\newcommand{\opsc}[5]{
\begin{scope}[shift={(#1)}]
	\draw[ \mthick, fill=#5, rounded corners=2pt] (-#2,-#3) rectangle (#2,#3);
	\draw (0,0) node {\scriptsize #4};
\end{scope}
}

\newcommand{\TCZ}[2]{
\begin{scope}[shift={(#1)}]
\def\disp{0.25};
        \def\rd{0.15};
        \def\rt{3.0};
        \def\rtt{5.5};
            \draw[Virtual,thick,rounded corners=2pt] (-\disp,\disp) -- (-\rt*\disp,\rt*\disp) -- (-\rtt*\disp,\rt*\disp);
            \draw[Virtual,thick,rounded corners=2pt] (\disp,\disp) -- (\rt*\disp,\rt*\disp) -- (\rtt*\disp,\rt*\disp);
            \draw[Virtual,thick,rounded corners=2pt] (-\disp,-\disp) -- (-\rt*\disp,-\rt*\disp) -- (-\rtt*\disp,-\rt*\disp);
            \draw[Virtual,thick,rounded corners=2pt] (\disp,-\disp) -- (\rt*\disp,-\rt*\disp) -- (\rtt*\disp,-\rt*\disp);
            \op{(-\disp,\disp)}{\rd}{\scriptsize $X$};
            \op{(\disp,\disp)}{\rd}{\scriptsize $X$};
            \op{(-\disp,-\disp)}{\rd}{\scriptsize $X$};
            \op{(\disp,-\disp)}{\rd}{\scriptsize $X$};
            \opsc{(2.3*\disp,0)}{0.12}{0.6}{}{lcolor};
            \opsc{(-2.3*\disp,0)}{0.12}{0.6}{}{lcolor};
            \ifnum#2=1
            \opsc{(0,2.3*\disp)}{0.6}{0.12}{}{lcolor};
            \opsc{(0,-2.3*\disp)}{0.6}{0.12}{}{lcolor};
            \fi
            \ifnum#2=2
            \opsc{(0,2.3*\disp)}{0.6}{0.12}{}{lcolor};
            \fi
            \ifnum#2=3
            \opsc{(0,-2.3*\disp)}{0.6}{0.12}{}{lcolor};
            \fi
\end{scope}
}

\newcommand{\GcTensor}[6]{
    \begin{scope}[shift={(#1)}]
    \ifnum#5=0
		\draw[Virtual] (-#2,0) -- (#2,0);
		\draw[] (0,#2) -- (0,-#2);
    \fi
    \ifnum#5=-1
		\draw[Virtual] (0,0) -- (#2,0);
		\draw[] (0,#2) -- (0,-#2);
    \fi
    \ifnum#5=1
		\draw[Virtual] (-#2,0) -- (0,0);
		\draw[] (0,#2) -- (0,-#2);
    \fi

    \ifnum#5=2
		\draw[Virtual] (-#2,0) -- (#2,0);
		\draw[] (0,0) -- (0,#2);
    \fi
    \ifnum#5=-2
		\draw[Virtual] (-#2,0) -- (#2,0);
		\draw[] (0,0) -- (0,-#2);
    \fi

    \ifnum#5=3
		\draw[Virtual] (-#2,0) -- (#2,0);
		\draw[] (0,-#2) -- (0,#2);
    \fi
    \ifnum#5=4
    \fi
        \draw[fill=#6, rounded corners=2pt,\mthick] (-#3,-#3) rectangle (#3,#3);
		\draw (0,0) node {\scriptsize #4};
	\end{scope}
}

\newcommand{\vTensor}[6]{
    \begin{scope}[shift={(#1)}]
        \draw[\mthick](-#4,0) -- (-#4,#5);
        \draw[\mthick](#4,0) -- (#4,#5);
        \draw[\mthick, fill=whitetensorcolor, rounded corners=2pt] (-#2,-#3) rectangle (#2, #3);
        \draw (0,0) node {\scriptsize #6};
        \draw[\mthick] (0,-#3)--(0,-#5);
    \end{scope}
}

\newcommand{\vtTensor}[6]{
    \begin{scope}[shift={(#1)}]
    \draw[\mthick](0,#5) -- (0,#3);
        \draw[\mthick](-#4,0) -- (-#4,-#5);
        \draw[\mthick](#4,0) -- (#4,-#5);
        \draw[\mthick, fill=whitetensorcolor, rounded corners=2pt] (-#2,-#3) rectangle (#2, #3);
        \draw (0,0) node {\scriptsize #6};
    \end{scope}
}

\newcommand{\xTensor}[5]{
    \begin{scope}[shift={(#1)}]
        \draw[Virtual] (-#2,0)--(-#4,0);
        \draw[\mthick, fill=whitetensorcolor, rounded corners=2pt] (-#2,-#3) rectangle (#2, #3);
        \draw (0,0) node {\scriptsize #5};
    \end{scope}
}

\newcommand{\xdTensor}[6]{
    \begin{scope}[shift={(#1)}]
        \draw[Virtual] (#2,#5)--(#4,#5);
        \draw[Virtual] (#2,-#5)--(#4,-#5);
        \draw[\mthick, fill=whitetensorcolor, rounded corners=2pt] (-#2,-#3) rectangle (#2, #3);
        \draw (0,0) node {\scriptsize #6};
    \end{scope}
}

\newcommand{\SideIdentityTensor}[4]{
	\begin{scope}[shift={(#1)}]
    \ifnum#4=-1
	   \draw [\mthick] (\doubledx-1,0.8) to  [bend right=90] (\doubledx-1,-0.8);
    \fi
    \ifnum#4=-2
	   \draw [\mthick] (\doubledx-1,0.8) to  [bend right=90] (\doubledx-1,-0.8);
      \draw [\mthick] (\doubledx-1,0.8) -- (\doubledx-0.5,0.8);
      \draw [\mthick] (\doubledx-1,-0.8) -- (\doubledx-0.5,-0.8);
    \fi
    \ifnum#4=-3
	   \draw [\mthick] (\doubledx-1,0.8) to  [bend right=90] (\doubledx-1,-0.8);
      \draw [\mthick] (\doubledx-1,0.8) -- (\doubledx-0.5,0.8);
      \draw [\mthick] (\doubledx-1,-0.8) -- (\doubledx-0.5,-0.8);
	\filldraw[color=black, fill=whitetensorcolor, \mthick] (\doubledx-1.4,0) circle (#3);
	\draw (\doubledx-1.4,0) node {#2};
    \fi
    \ifnum#4=1
	   \draw [\mthick] (-\doubledx+1,0.8) to  [bend left=90] (-\doubledx+1,-0.8);
    \fi
    \ifnum#4=2
	   \draw [\mthick] (-\doubledx+1,0.8) to  [bend left=90] (-\doubledx+1,-0.8);
      \draw [\mthick] (-\doubledx+1,0.8) -- (-\doubledx+0.5,0.8);
      \draw [\mthick] (-\doubledx+1,-0.8) -- (-\doubledx+0.5,-0.8);
    \fi
    \ifnum#4=3
	   \draw [\mthick] (-\doubledx+1,0.8) to  [bend left=90] (-\doubledx+1,-0.8);
      \draw [\mthick] (-\doubledx+1,0.8) -- (-\doubledx+0.5,0.8);
      \draw [\mthick] (-\doubledx+1,-0.8) -- (-\doubledx+0.5,-0.8);
	\filldraw[color=black, fill=whitetensorcolor, \mthick] (-\doubledx+1.4,0) circle (#3);
	\draw (-\doubledx+1.4,0) node {#2};
    \fi
\end{scope}
}

\newcommand{\whaMsymborg}[3]{
\begin{scope}[shift={(#1)}]
\filldraw[fill=#3] 
    (0,0) circle[radius=0.247];
    \node[anchor=center,scale=0.9] at (0,0) {#2};
\end{scope}
}

\newcommand{\whaMorg}[4]{
\begin{scope}[shift={(#1)}]
\def\rd{0.247};
\def\rlen{0.318};
\ifnum#2=1
    \draw[-mid] (0, \rd) -- (0, \rd+\rlen);
    \draw[-mid] (0, -\rd-\rlen) -- (0, -\rd);
    \draw[Virtual, -mid] (\rd+\rlen, 0) -- (\rd, 0);
    \draw[Virtual, -mid] (-\rd, 0) -- (-\rd-\rlen, 0);
\fi
\ifnum#2=2
    \draw[-mid] (0, \rd) -- (0, \rd+\rlen);
    \draw[-mid] (0, -\rd-\rlen) -- (0, -\rd);
    \draw[Virtual] (\rd+\rlen, 0) -- (\rd, 0);
    \draw[Virtual] (-\rd, 0) -- (-\rd-\rlen, 0);
\fi
\whaMsymborg{(0,0)}{#3}{#4};
\end{scope}
}

\newcommand{\whaM}[2]{
\whaMorg{#1}{#2}{$M$}{whampdocolor};
}

\newcommand{\whaATensor}[3]{
\begin{scope}[shift={(#1)}]
\def\rd{0.247};
\def\rlen{0.318};

\ifnum#2=1
\draw[] (0, \rd) -- (0, \rd+\rlen);
\draw[Virtual] (\rd+\rlen, 0) -- (\rd, 0);
\draw[Virtual] (-\rd, 0) -- (-\rd-\rlen, 0);
\fi
\ifnum#2=3
\draw[] (0, -\rd) -- (0, -\rd-\rlen);
\draw[Virtual] (\rd+\rlen, 0) -- (\rd, 0);
\draw[Virtual] (-\rd, 0) -- (-\rd-\rlen, 0);
\fi
\draw[fill=tensorcolor, rounded corners=2pt,\mthick] (-\rd,-\rd) rectangle (\rd,\rd);
\node[anchor=center] at (0,0) {#3};

\end{scope}
}

\newcommand{\mpo}[3]{
\begin{scope}[shift={(#1)}]
    \def\boxsizex{0.6}
    \def\boxsizey{0.5}
    \draw[Virtual] (2.223, 0.564) -- (1.729, 0.564);
    \draw[Virtual] (2.963, 0.564) -- (2.716, 0.564);
    \node[Virtual,anchor=center] at (3.212, 0.55) {$\cdots$};
    \draw[Virtual] (3.704, 0.564) -- (3.457, 0.564);
    \draw[Virtual] (4.198, 0.564) -- (4.691, 0.564);
    \draw[Virtual] (1.235, 0.564) arc[start angle=90, end angle=270, radius=0.2] -- (5.185, 0.164) arc[start angle=-90, end angle=90, radius=0.2];
    \draw[-mid] (1.482, 0.811) -- (1.482, 1.129);
    \draw[-mid] (2.469, 0.811) -- (2.469, 1.129);
    \draw[Virtual] (1.482, 0.564) -- (1.244, 0.564);
    \draw[-mid] (3.951, 0.811) -- (3.951, 1.129);
    \draw[Virtual] (4.198, 0.564) -- (3.951, 0.564);
    \draw[bevel, -mid] (1.482, 0) -- (1.482, 0.318);
    \draw[bevel, -mid] (2.469, 0) -- (2.469, 0.318);
    \draw[bevel, -mid] (3.951, 0) -- (3.951, 0.318);
    \whaMsymborg{(1.482, 0.564)}{\small #2}{white};
    \whaMsymborg{(2.469, 0.564)}{\small #2}{white};
    \whaMsymborg{(3.951, 0.564)}{\small #2}{white};
    \filldraw[ultra thin, fill=white] 
    (4.938-\boxsizex*0.5, 0.564-\boxsizey*0.5) rectangle ++(\boxsizex,\boxsizey);
    \node[anchor=center] at (4.938, 0.564) {\small #3};
\end{scope}
}

\newcommand{\mpotwo}[3]{
\begin{scope}[shift={(#1)}]
    \def\boxsizex{0.6}
    \def\boxsizey{0.5}
    \draw[Virtual] (2.964, 0.564) -- (1.729, 0.564);
    \draw[Virtual] (3.209, 0.564) -- (2.716, 0.564);
    \draw[Virtual] (1.235, 0.564) arc[start angle=90, end angle=270, radius=0.2] -- (3.211, 0.164) arc[start angle=-90, end angle=90, radius=0.2];
    \draw[-mid] (1.482, 0.811) -- (1.482, 1.129);
    \draw[-mid] (2.223, 0.811) -- (2.223, 1.129);
    \draw[Virtual] (1.482, 0.564) -- (1.244, 0.564);
    \draw[bevel, -mid] (1.482, 0) -- (1.482, 0.318);
    \draw[bevel, -mid] (2.223, 0) -- (2.223, 0.318);
    \whaMsymborg{(1.482, 0.564)}{\small #2}{white};
    \whaMsymborg{(2.223, 0.564)}{\small #2}{white};
    \filldraw[ultra thin, fill=white] 
    (2.964-\boxsizex*0.5, 0.564-\boxsizey*0.5) rectangle ++(\boxsizex,\boxsizey);
    \node[anchor=center] at (2.964, 0.564) {\small #3};
\end{scope}
}

\newcommand{\mpoone}[3]{
\begin{scope}[shift={(#1)}]
    \def\boxsizex{0.6}
    \def\boxsizey{0.5}
    \draw[Virtual] (3.209, 0.564) -- (2.716, 0.564);
    \draw[Virtual] (2.222, 0.564) arc[start angle=90, end angle=270, radius=0.2] -- (3.457, 0.164) arc[start angle=-90, end angle=90, radius=0.2];
    \draw[-mid] (2.469, 0.811) -- (2.469, 1.129);
    \draw[bevel, -mid] (2.469, 0) -- (2.469, 0.318);
    \whaMsymborg{(2.469, 0.564)}{\small #2}{white};
    \filldraw[ultra thin, fill=white] 
    (3.21-\boxsizex*0.5, 0.564-\boxsizey*0.5) rectangle ++(\boxsizex,\boxsizey);
    \node[anchor=center] at (3.21, 0.564) {\small #3};
\end{scope}
}

\usetikzlibrary{decorations.pathmorphing, decorations.markings, arrows, arrows.meta, shapes, patterns}
\tikzset{baseline={([yshift=-.5ex]current bounding box.center)}}
\tikzset{every path/.style={ line width=0.5pt, line cap=round }}
\colorlet{Virtual}{RedOrange}
\tikzstyle{bevel} = [ preaction = { draw, white, line width=3pt,  line cap = round } ]
\tikzstyle{bevel wide} = [ preaction = { draw, white, line width=4pt,  line cap = round } ]
\tikzstyle{symb} = [ draw=black, fill=black, line width=0.4pt, inner sep=1.5pt ]
\tikzstyle{mysymb} = [ draw=black, fill=white, circle, line width=0.3pt, inner sep=1pt, font=\small ] 
\tikzstyle{symb large} = [ inner sep=2.1pt ]
\tikzstyle{symb small} = [ inner sep=1pt   ]
\tikzstyle{symb tiny} = [ inner sep=0.8pt ]
\tikzstyle{symb fdisk} = [ circle ]
\tikzstyle{symb disk} = [ circle ]
\tikzstyle{symb square} = [ rectangle ]
\tikzstyle{symb fsquare} = [ rectangle ]
\tikzstyle{Msymb}=[draw=black, fill=whampdocolor, circle, inner sep=1pt, font=\small]
\tikzstyle{Nsymb}=[draw=black, fill=whampdocolorw, circle, inner sep=1pt, font=\small]
\tikzstyle{-mid} = [ decoration={ markings, mark = at position 0.50*\pgfdecoratedpathlength+0.6*3pt with \arrow{>[width=2pt]} }, postaction={decorate} ]
\tikzstyle{mid-} = [ decoration={ markings, mark = at position 0.50*\pgfdecoratedpathlength+0.6*3pt with \arrow{<[width=2pt]} }, postaction={decorate} ]

\newcommand\subsetsim{\mathrel{%
  \ooalign{\raise0.2ex\hbox{$\subset$}\cr\hidewidth\raise-0.8ex\hbox{\scalebox{0.9}{$\sim$}}\hidewidth\cr}}}

\newcommand{\ra}{\rightarrow}

\newcommand{\mS}{\mathcal{S}}
\newcommand{\mT}{\mathcal{T}}

\newcommand{\tr}{\mathrm{Tr}}
\newcommand{\bo}{\mathbbm{1}}

\newcommand{\rd}{\mathrm{d}}

\newcommand{\mA}{\mathcal{A}}
\newcommand{\dg}{\dagger}

\newcommand{\reg}{\text{reg}}
\newcommand{\basis}{\mathcal{B}}
\newcommand{\bdy}{\text{bdy}}
\newcommand{\CNOT}{\text{CNOT}}
\newcommand{\hbnt}{a}
\newcommand{\vbnt}{\mathbf{a}}
\newcommand{\Id}{\text{Id}}
\newcommand{\cu}{1_{\mA^*}}

\definecolor{XQ}{rgb}{1,0,0}

\definecolor{laur}{rgb}{0,0,1}

\definecolor{yuhan}{rgb}{0.9, 0, 0.5}

\allowdisplaybreaks

\begin{document}
\title{Trading Mathematical for Physical Simplicity:\\ Bialgebraic Structures in Matrix Product Operator Symmetries}

\author{Yuhan Liu}
\email{yuhan.liu@mpq.mpg.de}
\affiliation{Max Planck Institute of Quantum Optics, Hans-Kopfermann-Str. 1, Garching 85748, Germany}
\affiliation{Munich Center for Quantum Science and Technology (MCQST), Schellingstr. 4, 80799 M{\"{u}}nchen, Germany}
\author{Andras Molnar}
\affiliation{\mbox{University of Vienna, Faculty of Mathematics, Oskar-Morgenstern-Platz 1, 1090 Wien, Austria}}
\author{Xiao-Qi Sun}
\affiliation{Max Planck Institute of Quantum Optics, Hans-Kopfermann-Str. 1, Garching 85748, Germany}
\affiliation{Munich Center for Quantum Science and Technology (MCQST), Schellingstr. 4, 80799 M{\"{u}}nchen, Germany}
\author{\\Frank Verstraete}
\affiliation{Department of Applied Mathematics and Theoretical Physics, University of Cambridge,\\ Wilberforce Road, Cambridge, CB3 0WA, United Kingdom}
\affiliation{Department of Physics and Astronomy, Ghent University, Krijgslaan 281, 9000 Gent, Belgium}
\author{Kohtaro Kato}
\affiliation{Department of Mathematical Informatics, Graduate School of Informatics,\\ Nagoya University, Nagoya 464-0814, Japan}
\author{Laurens Lootens}
\email{ll708@cam.ac.uk}
\affiliation{Department of Applied Mathematics and Theoretical Physics, University of Cambridge,\\ Wilberforce Road, Cambridge, CB3 0WA, United Kingdom}


\begin{abstract}
Despite recent advances in the lattice representation theory of (generalized) symmetries, many simple quantum spin chains of physical interest are not included in the rigid framework of fusion categories and weak Hopf algebras. We demonstrate that this problem can be overcome by relaxing the requirements on the underlying algebraic structure, and show that general matrix product operator symmetries are described by a pre-bialgebra. As a guiding example, we focus on the anomalous $\mathbb Z_2$ symmetry of the XX model, which manifests the mixed anomaly between its $U(1)$ momentum and winding symmetry. We show how this anomaly is embedded into the non-semisimple corepresentation category, providing a novel mechanism for realizing such anomalous symmetries on the lattice. Additionally, the representation category which describes the renormalization properties is semisimple and semi-monoidal, which provides a new class of mixed state renormalization fixed points. Finally, we show that up to a quantum channel, this anomalous $\mathbb Z_2$ symmetry is equivalent to a more conventional MPO symmetry obtained on the boundary of a double semion model. In this way, our work provides a bridge between well-understood topological defect symmetries and those that arise in more realistic models.
\end{abstract}
\maketitle

\section{Introduction}

Symmetries play a vital role in understanding the quantum many-body problem, both when considering discrete lattice models and continuum quantum field theory (QFT). In fact, when trying to connect lattice and continuum approaches, imposing the correct symmetries on the lattice often provides the only guide to the desired continuum limit. An important aspect of symmetries in quantum many-body systems is their anomaly, which serves as an indicator of how the symmetry acts on the system. In QFT, anomalies provide obstructions to gauging the symmetry, i.e., turning the global symmetry action into a local one by adding gauge degrees of freedom. On the lattice, anomalies are often understood as the inability to realize the symmetry action in an on-site way, although a general definition of lattice anomalies is the subject of ongoing research \cite{Seifnashri:2025vhf,Tu:2025bqf,Shirley:2025yji}. One of the main drivers of these advancements is the problem of discretizing chiral fermions \cite{nielsen1981no,friedan1982proof}. These are well-known to possess anomalous symmetries, which for a long time were thought to directly prohibit their lattice regularization.

In more recent years, it was understood that anomalous symmetries can be realized on the lattice, provided that one allows them to act in a correlated manner on neighboring degrees of freedom. Such symmetries are naturally represented as matrix product operators (MPO) \cite{verstraete2004matrix,pirvu2010matrix,haegeman2017diagonalizing}, a type of tensor network that explicitly encodes the non-trivial entanglement structure of these symmetry operators \cite{buerschaper2014twisted,csahinouglu2021characterizing,bultinck2017anyons,williamson2016matrix}. A large class of MPO symmetries is obtained by considering weak Hopf algebras (WHA) and their representation categories \cite{lootens2021matrix,molnar2022matrix}, providing the lattice representation theory of (potentially non-invertible) fusion category symmetries \cite{bridgeman2023invertible}. A key advantage of this formal mathematical description is that, besides a classification of non-onsite MPO symmetries, it provides the necessary tools to gauge these symmetries \cite{lootens2023dualities} and exploit them in computational methods \cite{lootens2025entanglement}.

In this work, we investigate the underlying algebraic structure of an anomalous $\mathbb Z_2$ MPO symmetry of the Levin-Gu edge model \cite{levin2012braiding}, which is unitarily equivalent to the XX model. We show that this MPO symmetry goes beyond the framework provided by WHA, and that it can be understood as a representation of a non-semisimple non-counital pre-bialgebra. The corresponding representation categories are not fusion, and their associators allow for an explicit computation of the anomaly, as well as the renormalization properties of these symmetry operators. In turn, this leads us to a previously unknown family of renormalization fixed points of matrix product density operators (MPDO), possibly leading to new representations of topological order whose boundaries are these MPDOs. Finally, we show that this symmetry is related to an MPO symmetry described by a more conventional WHA by a quantum channel, meaning these symmetries are expected to coincide in the continuum limit. Our work provides a systematic method for extracting algebraic structures of MPO symmetries and bridges the gap between generic non-onsite symmetries of simple Hamiltonians and the better understood topological MPOs that often only arise at the boundaries of topological models.

\section{An anomalous $\mathbb Z_2$ symmetry}
We start by considering the Levin-Gu edge Hamiltonian \cite{levin2012braiding} on periodic boundary conditions:
\begin{equation}
    H^{(N)} = \sum_{i=1}^{N} X_i - Z_{i-1} X_i Z_{i+1},
\end{equation}
where $X,Y,Z$ denote the Pauli matrices. For $N = 0 \mod 4$, this model is unitarily equivalent to the XX model. The Hamiltonian $H^{(N)}$ is known to commute with the following quantum circuit $\mathbb Z_2$ symmetry \cite{chen2011two}
\begin{equation}
    U_{CZY}^{(N)} = \prod_{i=1}^{N} CZ_{i,i+1} \prod_{i=1}^{N} Z_i X_i.
\end{equation}
For $N$ even, this symmetry is obtained from the so-called momentum and winding $U(1)$ symmetries, respectively generated by $Q_M$ and $Q_W$ (see app. \ref{app:anomaly}), as $U_{CZY}^{(N)} = e^{i \pi Q_M} e^{i \pi Q_W}$. As we will see, $U_{CZY}$ is anomalous, which is a manifestation of the mixed anomaly between $U(1)_M$ and $U(1)_W$ \cite{chatterjee2025quantized,pace2025lattice}. Alternatively, $U_{CZY}$ can be written as a matrix product operator:
\begin{equation}
    U_{CZY}^{(N)} = \!\!\!\! \sum_{\{i\},\{j\}} \!\!\! \tr (A_1^{i_1j_1}\!\ldots A_1^{i_{N}j_{N}}) \ket{i_1 \ldots i_{N}}\!\bra{j_1 \ldots j_{N}}
\end{equation}
where the nonzero components of the rank-4 tensor $A_1$ are given by
\begin{equation}
    A_1^{01} =
    \begin{pmatrix}
        1 & 1 \\
        0 & 0
    \end{pmatrix}, \qquad
    A_1^{10} = 
    \begin{pmatrix}
        0 & 0 \\
        -1 & 1
    \end{pmatrix}.
\end{equation}
To construct an MPO algebra, we consider the product of two $U_{CZY}$ MPOs, leading to a bond dimension 4 MPO tensor $\tilde{A}_0$ defined as
\begin{equation}
    \tilde{A}_0^{00} = A_1^{01} \otimes A_1^{10}, \quad \tilde{A}_0^{11} = A_1^{10} \otimes A_1^{01},
\end{equation}
and all other components zero. Via a change of basis $X_{1,1}$ on the virtual degrees of freedom, these MPO tensors are equivalent to
\begin{align}
    X_{1,1} \tilde{A}_0^{00} (X_{1,1})^{-1} &= 
    \begin{pmatrix}
        0 & -1 & 1 \\
        0 & 1 & -1 \\
        0 & 0 & 0
    \end{pmatrix} 
    \oplus 0 =: A_0^{00} \oplus 0, \nonumber \\
    X_{1,1} \tilde{A}_0^{11} (X_{1,1})^{-1} &= 
    \begin{pmatrix}
        0 & 1 & 1 \\
        0 & 1 & 1 \\
        0 & 0 & 0
    \end{pmatrix}
    \oplus 0 =: A_0^{11} \oplus 0.
\end{align}
The MPO tensor $A_0$ is not injective, as the matrices $A_0^{00}$ and $A_0^{11}$ do not generate the full $3 \times 3$ matrix algebra. Despite this, it cannot be further decomposed into a direct sum of injective MPOs, and we will trace this property of the MPO back to the non-semisimplicity of the underlying algebraic structures. On periodic boundary conditions, the MPO generated by $A_0$ is just the identity, as expected from $(U_{CZY})^2 = \mathbb 1$.

To verify that this symmetry is anomalous, we consider the other possible products of the MPOs generated by $A_0$ and $A_1$, and obtain the similarity transformations $X_{a,b}$ with $a,b \in \mathbb Z_2$ satisfying
\begin{align}
    X_{0,1}\left(\sum_{j} A_0^{ij} \otimes A_1^{jk}\right) (X_{0,1})^{-1} &= A_1^{ik} \oplus 0_4, \nonumber\\
    X_{1,0}\left(\sum_{j} A_1^{ij} \otimes A_0^{jk}\right) (X_{1,0})^{-1} &= A_1^{ik} \oplus 0_4, \nonumber\\
    X_{0,0}\left(\sum_{j} A_0^{ij} \otimes A_0^{jk}\right) (X_{0,0})^{-1} &= A_0^{ik} \oplus 0_6,
    \label{eq:Xdef}
\end{align}
where $0_n$ denotes an $n \times n$ zero matrix; explicit expressions for $X_{a,b}$ can be found in app. \ref{app:fusion-tensor-expl}. We now define the matrices $Y_{a,b}$ and its right-inverse $(Y_{a,b})^{-1}$ by removing rows and columns of $X_{a,b}$ and $(X_{a,b})^{-1}$, respectively, such that
\begin{equation}
    Y_{a,b}\left(\sum_{j} A_a^{ij} \otimes A_b^{jk}\right) = A_{a+b}^{ik} Y_{a,b}.
    \label{eq:Ydef}
\end{equation}
The $Y_{a,b}$ are referred to as fusion tensors and implement the global multiplication property of these MPOs at the level of the local tensors; graphically, eq.~\eqref{eq:Ydef} is depicted as
\begin{equation}
    \begin{array}{c}
        \begin{tikzpicture}[scale=1,baseline={([yshift=-0.75ex] current bounding box.center)}]
            \def\rd{0.247};
            \def\rlen{0.318};
            \def\wth{0.25};
            \whaMorg{(0,0)}{2}{$A$}{white};
            \whaMorg{(0,\whady)}{2}{$A$}{white};
            \xTensor{(-0.765,\whady*0.5)}{0.25}{\whadx*0.6}{0.568}{$Y$};
            \draw[Virtual] (0.4,0.15+\whady) node {\scriptsize $a$};
            \draw[Virtual] (0.4,0.15) node {\scriptsize $b$};
            \draw[Virtual] (-1.4,0.15+\whady*0.5) node {\scriptsize $a\!+\!b$};
        \end{tikzpicture}
    \end{array}
    =
    \begin{array}{c}
        \begin{tikzpicture}[scale=1,baseline={([yshift=-0.75ex] current bounding box.center)}]
            \def\rd{0.247};
            \def\rlen{0.318};
            \def\wth{0.25};
            \whaMorg{(0,\whady*0.5)}{2}{$A$}{white};
            \xdTensor{(0.765,\whady*0.5)}{0.25}{\whadx*0.6}{0.568}{\whady*0.5}{$Y$};
            \draw[Virtual] (1.2,0.15+\whady) node {\scriptsize $a$};
            \draw[Virtual] (1.2,0.15) node {\scriptsize $b$};
            \draw[Virtual] (-0.6,0.15+\whady*0.5) node {\scriptsize $a\!+\!b$};
        \end{tikzpicture}
    \end{array}.
\end{equation}
It is well understood that the failure of this local multiplication to be associative provides an obstruction to making this symmetry on-site \cite{Seifnashri:2025vhf}, indicating an anomaly. To quantify this, one defines an associator $\omega \in H^3(\mathbb Z_2,U(1)) \simeq \mathbb Z_2$ as
\begin{equation}
    Y_{a+b,c} (Y_{a,b} \otimes \mathbb 1_c) = \omega(a,b,c) Y_{a,b+c} (\mathbb 1_a \otimes Y_{b,c}).
\end{equation}
Using the explicit expressions for $Y_{a,b}$ (see app. \ref{app:fusion-tensor-expl}), we find
\begin{equation}
    \omega(a,b,c) =
    \begin{cases}
        -1, \quad &a = b = c = 1,\\
        1, \quad &\text{otherwise}.
    \end{cases}
\end{equation}
This is indeed a representative of the non-trivial element in $H^3(\mathbb Z_2,U(1))$, and as such, this $\mathbb Z_2$ symmetry is anomalous. In \cite{chen2011two}, a similar computation is performed, but there, the off-diagonal blocks in $A_0$ are projected out by hand, which obscures the underlying algebraic structures as in that case eq. ~\eqref{eq:Ydef} is not satisfied.

\begin{table*}
    \begin{tabular}{|c|c|c|c|c|c|c|}
    \hline
        Indecomposable module &  dimension & simple & projective & basis & $\mathrm{rad}(P)$ & $P/\mathrm{rad}(P)$\\
        \hline
      $P_0$ & 3 & no & yes & $\{e^2,e^4,e^0-e^3-e^5-e^8\}$  & $P_1$ & $S_0$\\
      $P_1$   & 2  & no & yes & $\{e^1,e^3\}$ & $S_0$ & $S_1$\\
      $P_2 (=S_2) \cong P_3 $ & 2 & yes &  yes & $\{e^5,e^7\}$, $\{e^6,e^8\}$& $\emptyset$ & $S_2$\\
      \hline
      $S_0$ & 1 & yes & no & - & - &-\\
      $S_1$ & 1 & yes & no & - & -& -\\
    \hline
    \end{tabular}
    \caption{Indecomposable modules $P_a$ from the regular module of the unitized algebra $\mA^*_+$, and the simple modules $S_a$ of which they are the projective covers. Here $\{e^1,\cdots,e^8\}$ relabels $\{e^0_{12},e^0_{13},e^0_{22},e^0_{23}, e^1_{11},e^1_{12},e^1_{21},e^1_{22}\}$ for simplicity of notation, and $e^0$ is the unit element of $\mA^*_+$. }
    \label{tbl:modules}
\end{table*}

\section{MPO representations of bialgebras}
\label{sec:algebra_coalgebra}
In this section, we derive the algebraic structures underlying the MPOs generated by $A_0$ and $A_1$, and show that they go beyond the current classification scheme for MPO symmetries in terms of  (weak) Hopf algebras and their representation categories. To this end, we consider the MPOs generated by the tensors $A_a$ with the additional insertion of a boundary condition $B_a$:
\begin{equation} \label{eq:mpo_multiplication}
    O^{(N)}(B_a) \, = \,
    \begin{tikzpicture}[scale=1,baseline={([yshift=-0.75ex] current bounding box.center)}]
        \mpo{(0,0)}{$A$}{$B_a$};
        \draw[Virtual] (4.4,0.7) node {\scriptsize $a$};
    \end{tikzpicture}
\end{equation}
The boundary conditions $B_0 \in \{e_0^{12},e_0^{13},e_0^{22},e_0^{23}\}$ and $B_1 \in \{e_1^{11},e_1^{12},e_1^{21},e_1^{22}\}$  lead to a maximal set of linearly independent MPOs for $N\geq 2$, where the matrices $e_0^{mn}$ and $e_1^{mn}$ are respectively defined as $3 \times 3$ and $2 \times 2$ matrix units with $e_a^{mn}=|n\rangle\langle m|$ (notice the transposition w.r.t.\ the usual notation) such that $O^{(1)}(e_a^{mn})=(A_a)_{mn}$. 

By virtue of eqs.~\eqref{eq:Xdef} and \eqref{eq:Ydef}, the MPOs $O^{(N)}(B_a)$ form a closed algebra under multiplication
\begin{align}
\label{eq:mpo_mult}
    &\hspace{2em}\begin{array}{c}
        \begin{tikzpicture}[scale=1,baseline={([yshift=-0.75ex] current bounding box.center)}]
            \mpo{(0,0)}{$A$}{$e_a^{mn}$};
            \mpo{(0,-0.812)}{$A$}{$e_b^{pq}$};
            \draw[Virtual] (4.4,0.7) node {\scriptsize $a$};
            \draw[Virtual] (4.4,0.7-\whady) node {\scriptsize $b$};
        \end{tikzpicture}
    \end{array}\\
    &\quad=\sum_{c,rs} \lambda_{(a,mn)(b,pq)}^{(c,rs)}
    \begin{array}{c}
        \begin{tikzpicture}[scale=1,baseline={([yshift=-0.75ex] current bounding box.center)}]
            \mpo{(0,0)}{$A$}{$e_c^{rs}$};
            \draw[Virtual] (4.4,0.7) node {\scriptsize $c$};
        \end{tikzpicture}
    \end{array}\nonumber
\end{align}
for every system size $N$. Therefore $O^{(N)}$ is a representation of an 8-dimensional algebra $\mathcal{A}$, a basis of which we denote the same as the corresponding boundary conditions $e_{a}^{mn}$: as a vector space,
\begin{equation}
    \mathcal{A} = \text{Span} \{ e_0^{12},e_0^{13},e_0^{22},e_0^{23}, e_1^{11},e_1^{12},e_1^{21},e_1^{22}\}.
\end{equation}
The multiplication $\lambda$ of $\mathcal{A}$ follows from eq.~\eqref{eq:Ydef}, and is given by
\begin{equation}
\label{eqn:multiplication-lambda}
        e_a^{mn} \cdot e_b^{pq} = \sum_{rs} [Y_{a,b}(e_a^{mn} \otimes e_b^{pq})Y_{a,b}^{-1}]_{sr}  e_{a+b}^{rs},
\end{equation}
which is written out in app.~\ref{app:explicit-mul-basis-trans} using the explicit expressions of $Y_{a,b}$. This algebra $\mathcal A$ is in fact a twisted group algebra of $\mathbb Z_2 \times \mathbb Z_2 \times \mathbb Z_2$ isomorphic to the direct sum of two 2-dimensional matrix algebras $\mathcal M_2 \oplus \mathcal M_2$, which is semisimple (note that the basis $\{e_a^{mn}\}$ is not the basis in which this structure is apparent; see app.~\ref{app:explicit-mul-basis-trans} for a basis rotation). We label the two irreducible representations of $\mA$ as $\phi_\vbnt$ with $\vbnt=1,2$.

Importantly, the MPO representation also defines an associative linear operation $\Delta: \mathcal{A}\to \mathcal{A} \otimes \mathcal{A}$, called comultiplication, by $\big(O^{(l_1)}\otimes O^{(l_2)}\big)\circ\Delta(e_a^{mn}):=O^{(l_1+l_2)}(e_a^{mn})$ for any $l_1,l_2$; meaning that it provides a relation between MPO representations on different sizes. In our basis $e_a^{mn}$ for the algebra $\mA$ we have the following identity:
\begin{equation*}
    \sum_{p}\begin{array}{c}
        \begin{tikzpicture}[scale=1,baseline={([yshift=-0.75ex] current bounding box.center)}]
            \mpoone{(0,0)}{$A$}{$e_a^{mp}$};
            \mpoone{(1.8,0)}{$A$}{$e_a^{pn}$};
            \draw[Virtual] (4.6,0.7) node {\scriptsize $a$};
            \draw[Virtual] (2.8,0.7) node {\scriptsize $a$};
        \end{tikzpicture}
    \end{array} = 
    \begin{array}{c}
        \begin{tikzpicture}[scale=1,baseline={([yshift=-0.75ex] current bounding box.center)}]
            \mpotwo{(0,0)}{$A$}{$e_a^{mn}$};
            \draw[Virtual] (2.55,0.7) node {\scriptsize $a$};
        \end{tikzpicture}
    \end{array},
    \label{eq:mpo_comultiplication}
\end{equation*} 
from which we can read off the comultiplication as
\begin{equation}
\label{eqn:A-comultiplication}
    \Delta(e_a^{mn}) = \sum_p e_a^{mp} \otimes e_a^{pn},
\end{equation}
where in both equations the sum $p$ only takes values for which $e_{a}^{mp}$ and $e_a^{pn}$ are both valid basis elements; for example, $\Delta(e_0^{22}) = e_0^{22} \otimes e_0^{22}$. 
The fact that the structure constants for the algebra in eq.~\eqref{eq:mpo_mult} do not depend on the system size $N$ guarantees that the comultiplication is compatible with the multiplication in $\mA$, meaning that $(\Delta \otimes \Id) \circ \Delta = (\Id\otimes \Delta ) \circ \Delta$ and $\Delta(xy) = \Delta(x) \Delta(y)$ for $x,y\in\mA$. This turns $\mA$ into a pre-bialgebra~\cite{molnar2022matrix} (see app.~\ref{app:intro-algebra} for formal definition). While we used a particular example for illustration, this method provides a general way for extracting the pre-algebraic structure of any consistent MPO symmetry.

As a result, the dual algebra $\mathcal A^*$ has a basis $\{e^a_{mn}\}$, with multiplication structure
\begin{equation}
\label{eqn:mul-As}
    e^a_{mn} e^b_{pq} = \delta_{ab} \delta_{np} e^a_{mq}.
\end{equation}
 That is, the defining representation of $\mathcal{A}^*$ is
\begin{equation*}
    \mathcal{A}^* = \left\{ \begin{pmatrix}
    .&a_1&a_2 &. &.\\
        .&a_3&a_4 &. & .\\
        .&.&.&.&.\\
         . &.&. & a_5 & a_6 \\
        .&.&. & a_7&a_8
    \end{pmatrix} \middle | a_1,\dots, a_8 \in \mathbb{C} \right\}.
\end{equation*}
This algebra does not have a unit, and is non-semisimple. The absence of the unit of $\mA^*$ implies the absence of the monoidal unit of $\mathrm{Rep}(\mA)$, as discussed later. The representation theory of $\mA^*$ is more involved. As a non-semisimple algebra, $\mA^*$ has reducible but indecomposable modules, i.e. it contains proper submodules but cannot be written as a direct sum of them. To obtain all its simple and indecomposable projective modules, one can unitize $\mA^*$ to obtain a unital algebra $\mA^*_+$ and decompose the regular module $m_\reg$ of $\mA^*_+$ to a direct sum of projective indecomposable modules $P_\hbnt$, 
\begin{equation}
m_\reg \cong P_1 \oplus P_2 \oplus P_3 \oplus P_0,    
\end{equation}
where $P_2\cong P_3$ is 2-dim, $P_1$ is 2-dim and $P_0$ is 3-dim. The properties are derived in app.~\ref{app:rep} and summarized in Table~\ref{tbl:modules}.  
The three simple module $S_\hbnt$ of $\mA^*_+$ can then be obtained by quotienting $P_\hbnt$ by its radical $\text{rad}(P_\hbnt)$. We label the irreducible representation corresponding to the simple module $S_\hbnt$ as $\psi_{S_\hbnt}$.

To study the representation category of $\mA$, we need its comultiplication structure to allow taking tensor products of representations. Given $\phi_\vbnt,\phi_{\mathbf{b}}\in\text{Irr}(\mA)$ where $\mA$ is semisimple, tensor product representations can be decomposed into a direct sum of irreducible representations,
\begin{equation}
\label{eqn:tensor-prod-rep-main}
    \phi_\vbnt \boxtimes \phi_{\mathbf{b}}:=(\phi_\vbnt\otimes \phi_{\mathbf{b}})\circ\Delta_\mA \simeq \bigoplus_{\mathbf{c}\in\text{Irr}(\mA)} \bo_{N_{\mathbf{ab}}^{\mathbf{c}}} \otimes \phi_{\mathbf{c}}.
\end{equation}
The non-negative integers $N_{\mathbf{ab}}^{\mathbf{c}}$ encode the fusion rules of the category $\mathcal{D}=\text{Rep}(\mA)$, which are
\begin{equation}
\phi_{\mathbf{a}} \boxtimes \phi_{\mathbf{b}} \simeq \phi_{\mathbf{1}} \oplus \phi_{\mathbf{2}}, \quad \forall \mathbf a, \mathbf b.
\end{equation}
The absence of the monoidal unit in $\text{Rep}(\mA)$ is a consequence of the absence of unit in $\mA^*$. As such, this is not a monoidal category, but rather is referred to as a semi-monoidal category \cite{kock2008elementary}.

Similarly, the multiplication structure on $\mA$ defines a comultiplication structure on $\mA^*$, which allows taking tensor products of representations of $\mA^*$ (as an algebra). Since $\mA^*$ is non-semisimple, the corresponding representation category $\mathcal C = \text{Rep}(\mA^*)$ is non-semisimple as well. Taking the indecomposable representations $\psi_{P_0}, \psi_{P_1}, \psi_{P_2}$ together with the irreducible representations $\psi_{S_0}, \psi_{S_1}$, we find the following fusion rules:
\begin{align*}
    \psi_{P_0} &\boxtimes \psi_{P_0} \simeq \psi_{P_0} \oplus \bo_6 \otimes \psi_{S_0},\\
    \psi_{P_1} &\boxtimes \psi_{P_1} \simeq \psi_{P_1} \oplus \bo_2 \otimes \psi_{S_0},\\
    \psi_{P_2} &\boxtimes \psi_{P_2} \simeq \psi_{P_0} \oplus \psi_{S_0},\\
    \psi_{P_0} &\boxtimes \psi_{P_1} \simeq \psi_{P_1} \boxtimes \psi_{P_0} \simeq \psi_{P_0} \oplus \bo_3 \otimes \psi_{S_0},\\
    \psi_{P_0} &\boxtimes \psi_{P_2} \simeq \psi_{P_2} \boxtimes \psi_{P_0} \simeq \psi_{P_2} \oplus \bo_4 \otimes \psi_{S_0},\\
    \psi_{P_1} &\boxtimes \psi_{P_2} \simeq \psi_{P_2} \boxtimes \psi_{P_1} \simeq \psi_{P_2} \oplus \bo_2 \otimes \psi_{S_0},\\
    \psi &\boxtimes \psi_{S_0} \simeq \psi_{S_0} \boxtimes \psi \simeq \bo_{\text{dim}(\psi)} \otimes \psi_{S_0}, \quad \forall \psi,\\
    \psi &\boxtimes \psi_{S_1} \simeq \psi_{S_1} \boxtimes \psi \simeq \psi, \quad \forall \psi.
\end{align*}
This is a non-semisimple monoidal category, with unit $\psi_{S_1}$ and associators that can be computed as before. One possible semisimplification of $\mathcal C$ is obtained by restricting to the indecomposable representations $\psi_{P_0}$ and $\psi_{P_2}$ which as shown above yields the semion category $\text{Vec}_{\mathbb Z_2}^\omega$, i.e. the fusion category of $\mathbb Z_2$ graded vector spaces with nontrivial associator $\omega$ \cite{etingof2015tensor}.

\section{MPDO RG fixed points}
\label{sec:mpo-reconstruction}
Conversely to the approach taken above, the algebraic structure of $\mA$ and knowledge of representation lead to the construction of the MPO tensors, and in particular, the construction of renormalization fixed points of MPDO. This is significant because it goes beyond the current framework of MPDO fixed point construction based on $C^*$-weak Hopf algebra. We speculate that this implies new representations of topological order whose boundaries are this new family of MPDO RFPs. 

The MPO tensor construction utilizes the representation $\psi$ of $\mA^*$ and the representation $\phi$ of $\mA$~\cite{molnar2022matrix},
\begin{equation}
\label{eqn:Ma-def}
\begin{array}{c}
        \begin{tikzpicture}[scale=1.,baseline={([yshift=-0.65ex] current bounding box.center)}]
		\draw (-0.75,0) node {$\alpha$};
		\draw (0.75,0) node {$\beta$};
		\draw (0,0.75) node {$i$};
        \draw (0,-0.75) node {$j$};
        \whaMorg{(0,0)}{2}{\small $A$}{white};
        \draw[Virtual] (0.5,0.15) node {\scriptsize $\psi$};
        \draw[] (-0.15,0.5) node {\scriptsize $\phi$};
        \end{tikzpicture}
        \end{array}= \sum_{I\in\basis} [\phi(e_I)]_{ij} [\psi(e^I)]_{\alpha\beta}.
\end{equation}
In our example, taking $\phi=\phi_{\mathbf{1}}$ with $\psi=\psi_{S_2}=\psi_{P_2}$ generates precisely the tensor $A_1$ defined in the previous section, and with $\psi=\psi_{P_0}$ generates the tensor $A_0$. 

A tensor $M$
generates a valid matrix-product density operator (MPDO) for system size $N$
\begin{equation}
\begin{aligned}
     \rho^{(N)}(M)=&\sum_{\lbrace i,j\rbrace}\tr\left( M^{i_1 j_1} M^{i_2 j_2}\cdots M^{i_N j_N}\right)\\
     &\quad  |i_1 i_2 \cdots i_N\rangle\langle j_1 j_2 \cdots j_N|,
\end{aligned}
\end{equation}
if $\rho^{(N)}=(\rho^{(N)})^\dg\geq 0$~\cite{verstraete2004matrix,zwolak2004mixed,cirac2017matrix}. A tensor $M$ generating an MPDO is called a renormalization fixed point (RFP) if there exist two quantum channels, $\mT$ and $\mS$, that act on the physical space and fulfill the condition~\cite{cirac2017matrix,kato2024exact,liu2025parent}
\begin{align}
    \label{eqn:RFP-def-main}
    \raisebox{1.3ex}{
        \begin{tikzpicture}[scale=1.,baseline={([yshift=-0.75ex] current bounding box.center)}
        ]
       \whaM{(0,0)}{2};
        \draw [->, thick](1.,0.4) to [out=45,in=135] (3*\whadx-1.,0.4);
        \draw (1.5*\whadx,0.85) node {$\mT$};
        \draw [->, thick](3*\whadx-1,-0.4) to [out=225,in=-45] (1.,-0.4);
        \draw (1.5*\whadx,-0.32) node {$\mS$};
        \end{tikzpicture}
        }
        \begin{tikzpicture}[scale=1.,baseline={([yshift=-0.75ex] current bounding box.center)}
        ]
         \whaM{(3*\whadx,0)}{2};
        \whaM{(4*\whadx,0)}{2};
        \end{tikzpicture}
        \,.
    \end{align}
To generate the MPDO renormalization fixed point tensor from the pre-bialgebra $\mA$, we define MPO tensors $\{\tilde{M}_\vbnt\}$ 
where $\tilde{M}_\vbnt$ is constructed as eq.~\eqref{eqn:Ma-def} with $\phi=\phi_{\vbnt}\in\mathrm{Irr}(\mA)$ and $\psi$ a faithful representation of $\mA^*$. The fixed-point tensor $M$ is constructed by taking a suitable superposition of $\tilde{M}_\vbnt$. 

\begin{thm}
\label{prop:sufficient-main}
    Let $\mA$ be an associative semisimple $C^*$-pre-bialgebra, with $\mA^*$ possibly lacking a unit and not necessarily semisimple. If (1) the fusion multiplicities $N^{\mathbf{a}}_{\mathbf{bc}}$ 
  are transitive, and (2) for any $\vbnt\in\mathrm{Irr}(\mA)$ there exists $\vbnt^*\in\mathrm{Irr}(\mA)$ such that $N_{\vbnt^*}=N_\vbnt^T$ where $(N_\vbnt)_{\mathbf{cb}}=N^{\mathbf{c}}_{\mathbf{ab}}$, then the matrix product operator
    \begin{equation}
    \label{eqn:fixed-point-generate-main}
    \begin{array}{c}
        \begin{tikzpicture}[scale=1.,baseline={([yshift=-0.65ex] current bounding box.center)}]
        \whaM{(0,0)}{2};
        \end{tikzpicture}
        \end{array}
       =\bigoplus_{\vbnt\in\mathrm{Irr}(\mA)}\frac{d_\vbnt}{\mathrm{FPdim}(\mathcal D)} \begin{array}{c}
        \begin{tikzpicture}[scale=1.,baseline={([yshift=-0.65ex] current bounding box.center)}]
        \whaMorg{(0,0)}{2}{$A$}{white};
        \draw[] (-0.2,0.4) node {\scriptsize $\phi_\vbnt$};
        \end{tikzpicture}
        \end{array}
    \end{equation}
    with $d_\vbnt$ being the spectral radius of matrix $N_\vbnt$ and $\mathrm{FPdim}(\mathcal D):=\sum_\vbnt d_\vbnt^2$, satisfies the renormalization fixed point conditions given in eq.~\eqref{eqn:RFP-def-main}. 
\end{thm}
We leave the proof in app.~\ref{app:proof-suff} that utilizes the vertical canonical form of the MPDO tensor. Applying this general result to our example with faithful representation $\psi=\psi_{P_0}\oplus \psi_{S_2}$ leads to the fixed-point tensor $M$ with physical dimension 4 and bond dimension 5 that generates the non-trivial density matrix~\cite{lessa2025mixed}
\begin{equation}
    \rho_{CZY}^{(2N)}=\frac{1}{2^{2N}}(\bo_2^{\otimes 2 N}+U_{CZY}^{(2N)}) 
\end{equation}
up to local unitaries. 
With the periodic boundary condition, $\psi_{S_1}$ effectively generate the same state as $\psi_{P_0}$ and we can as well choose  $\psi=\psi_{S_1}\oplus \psi_{S_2}$, which leads to the fixed-point tensor $M$ with smaller bond dimension 3 and still generates $\rho_{CZY}^{(N)}$. 

Reversely, starting from the fixed-point tensor $M$ of $\rho_{CZY}^{(N)}$, one can also reconstruct the bialgebraic structure of $\mA$ by bringing the tensor $M$ into vertical canonical form.

\section{Relation to the double semion model}
\label{sec:relation}
A related density matrix is the boundary of the double semion model~\cite{xu2018tensor} $\rho^{(2N)}_{\bdy}$, which can be obtained from $\rho_{CZY}^{(N)}$ by local quantum channels,
\begin{equation}\label{eq:CZX to semion}
    \rho_{\bdy}^{(2N)}=\mathcal{E}^{\otimes N}(u^{\otimes N} \rho_{CZY}^{(N)} (u^\dg)^{\otimes N})
\end{equation}
with $\mathcal{E}(\rho)=\CNOT (\rho\otimes|0\rangle\langle 0|)\CNOT$ and single-site unitary $u=e^{i(\pi/4)Z}$~\footnote{Here, $\rho_{\bdy}^{(2N)}=\mathcal{E}^{\otimes N}(\rho_{CZX}^{(N)})$ where $\rho_{CZX}^{(N)}=\prod_{i=1}^{N} CZ_{i,i+1} \prod_{i=1}^{N} X_i$ and $u^{\otimes N}$ is the unitary that relates $\rho_{CZX}$ and $\rho_{CZY}$. Strictly speaking, the mapping from $\rho_{CZX}$ to $\rho_{CZY}$ is exact only for $N=0 \mod 4$. For $N=0\mod 2$, we can use $(u\otimes u^\dg)^{\otimes N/2}$}. Vice versa, $\rho_{CZY}^{(N)}$ can be recovered from $\rho_{\bdy}^{(2N)}$ via
\begin{equation}
\begin{aligned}
    \rho_{CZY}^{(N)} &= (u^\dg)^{\otimes N}\left(\mathcal{R}^{\otimes N}(\rho_{\bdy}^{(2N)})\right) u^{\otimes N},
\end{aligned}
\end{equation}
with $ \mathcal{R}(\rho) = \tr_2  [\CNOT (\rho) \CNOT]$, where $\tr_2$ denotes tracing out the second site. Using the MPDO tensor $M_\bdy$ that generates $\rho_{\bdy}^{(2N)}$ and brings it into the vertical canonical form, one can reconstruct a bialgebraic structure $\mA_\bdy$. The multiplication structure of $\mA_\bdy$ is the same as $\mA$, while the comultiplication structures are different. It turns out that the representation category of $\mA_\bdy$ is $\text{Vec}_{\mathbb{Z}_2}^\omega$ and $\mA_\bdy$ is a $C^*$-weak Hopf algebra.  

\section{Discussion and outlook}
We have demonstrated that, given an MPO symmetry, one can systematically extract its underlying bialgebraic structures. By applying this to the anomalous $\mathbb Z_2$ symmetry considered in this work, we find that it lies outside the known framework of weak Hopf algebras and their representation categories. In particular, the monoidal category that describes the symmetry operators is not semisimple, and we recover the anomaly as an associator upon semisimplification. Conversely, the category of representations of the MPO algebra itself is a semi-monoidal category due to the lack of a monoidal unit. While we focus on a particular example, our methodology can be readily applied to arbitrary MPO symmetries. In particular, in app.~\ref{sec:finite group generalization} we show how our approach generalizes to the finite group case. By uncovering the algebraic structures underlying MPO symmetries, we open the door to applying the formal machinery for gauging these symmetries, developed for ordinary fusion category symmetries, to more general symmetries often encountered in physically relevant Hamiltonians. 

Furthermore, identifying this algebraic structure beyond weak Hopf algebras leads to the construction of a previously unknown family of MPDO fixed points. This prompts the question: do \textit{all} MPDO RFPs satisfy the conditions of Theorem~\ref{prop:sufficient-main}? Addressing this would yield a necessary and sufficient condition characterizing the algebraic structure of all MPDO RFPs. Another important question follows: can we classify the algebraic structures underlying MPDO RFPs in relation to the classification of mixed-state quantum phases~\cite{coser2019classification,de2022symmetry,ma2023average,sang2024mixed,ellison2024towards,sohal2025noisy,sun2025anomalous}? Some results are already known -- for instance, MPDO RFPs constructed from $C^*$-Hopf algebras necessarily belong to the trivial phase~\cite{ruiz2024matrix}. We leave the complete classification to future work. 

\section*{Acknowledgement}
We thank Arkya Chatterjee, Ignacio Cirac, Clement Delcamp, Marta Florido Llinàs, José Garre-Rubio, David Pérez-García, and Zhiyuan Wang for their insightful discussions. Y.L. and X.-Q.S. are supported by the Alexander von Humboldt Foundation. L.L. is supported by an EPSRC Postdoctoral Fellowship (grant No. EP/Y020456/1). K.~K. acknowledges support from MEXT-JSPS Grant-in-Aid for Transformative Research Areas (B), No. 24H00829; from JSPS KAKENHI Grant No. 23K17668. A.M. acknowledges support by the Austrian Science Fund(FWF) via Grants 10.55776/COE1, by the European Union – NextGenerationEU, and by the European Union’s Horizon 2020 research and innovation programme through Grant No. 863476”. F.V. is supported by EOS (grant No. 40007526), IBOF (grant No. IBOF23/064), and BOF- GOA (grant No. BOF23/GOA/021). The authors thank the Yukawa Institute for Theoretical Physics at Kyoto University, where part of this work was done during the YITP-I-25-02 on "Recent Developments and Challenges in Tensor Networks: Algorithms, Applications to Science, and Rigorous Theories."

\appendix

\section{U(1) symmetries of the Levin-Gu and XX model}
\label{app:anomaly}
For $N = 0 \mod 2$, the Levin-Gu Hamiltonian
\begin{equation}
    H_{\text{LG}} = \sum_{i=1}^{N} X_i - Z_{i-1} X_i Z_{i+1}.
\end{equation}
is equivalent to
\begin{equation}
    H_2 := U^\dag H_{\text{LG}} U = \sum_{i=1}^{N} X_i Z_{i+1} - Z_{i} X_{i+1}
\end{equation}
with
\begin{equation}
    U = \prod_{i=1}^{N/2} CZ_{2i,2i+1} X_{2i},
\end{equation}
which after blocking is an on-site unitary transformation. Furthermore, when $N = 0 \mod 4$, $H_2$ is equivalent to the XX model
\begin{equation}
    H_{\text{XX}} = \tilde U^\dag H_2 \tilde U = \sum_{i=1}^{N} X_i X_{i+1} + Z_i Z_{i+1}
\end{equation}
with
\begin{equation}
    \tilde U = \prod_{i=1}^{N/2} H_{2i} \prod_{i=1}^{N/4} X_{4i+1} X_{4i+2}.
\end{equation}
$H_2$ commutes with
\begin{equation}
    Q_M = \frac{1}{2} \sum_{i=1}^{N} Y_i, \quad Q_W =  \frac{N}{4} + \frac{1}{4} \sum_{i=1}^{N} Z_i Z_{i+1},
\end{equation}
so $H_{\text{LG}}$ commutes with
\begin{align}
    \tilde Q_M &:= U Q_M U^\dag = \frac{1}{2} \sum_{i=1}^{N/2} Z_{2i} Y_{2i+1} - Y_{2i} Z_{2i+1},\\
    \tilde Q_W &:= U Q_W U^\dag = \frac{N}{4} - \frac{1}{4} \sum_{i=1}^{N} Z_i Z_{i+1}.
\end{align}
These charges generate a compact $U(1)$ symmetry. Looking at the $\mathbb Z_2$ subgroups of these $U(1)$ symmetries, we find
\begin{equation}
    e^{i\pi \tilde Q_M} = \prod_{i=1}^N X_i, \quad e^{i\pi \tilde Q_W} = \prod_{i=1}^N CZ_{i,i+1} Z_i,
\end{equation}
and so
\begin{equation}
    U_{CZY} := \prod_{i=1}^N CZ_{i,i+1} Z_i X_i = e^{i\pi \tilde Q_W} e^{i\pi \tilde Q_M}. 
\end{equation}
We note that while $\tilde{Q}_M$ and $\tilde{Q}_W$ do not commute, the $\mathbb{Z}_2$ operators they generate commute $[e^{i\pi \tilde Q_M},e^{i\pi \tilde Q_W}]=0$. 

\section{Explicit form of fusion tensors}
\label{app:fusion-tensor-expl}
The explicit form of matrices $X_{a,b}$ are:
\begin{equation}
    X_{1,1} = \frac{1}{\sqrt{2}}
    \begin{pmatrix}
        0 & 1 & 1 & 0 \\
        0 & -1 & 1 & 0 \\
        -1 & 0 & 0 & 1 \\
        -1 & 0 & 0 & -1 
    \end{pmatrix},
\end{equation}
and
\begin{equation}
    X_{0,1} =
    \begin{pmatrix}
        0 & 0 & 1 & 0 & 0 & -1 \\
        0 & 0 & 0 & 1 & -1 & 0 \\
        0 & 1 & 0 & -1 & 0 & 0 \\
        -1 & 0 & -1 & 0 & 0 & -1 \\
        0 & -1 & 0 & 1 & -2 & 0 \\
        -1 & 0 & -1 & 0 & 0 & 2
    \end{pmatrix},
\end{equation}
and
\begin{equation}
    X_{1,0} =
    \begin{pmatrix}
        0 & 1 & 0 & 0 & 0 & 1 \\
        0 & 0 & 1 & 0 & 1 & 0 \\
        1 & -1 & 0 & 0 & 0 & 1 \\
        0 & 0 & 1 & -1 & -1 & 0 \\
        1 & -1 & 0 & 0 & 0 & -2 \\
        0 & 0 & 2 & 1 & 1 & 0
    \end{pmatrix},
\end{equation}
and
\begin{equation}
    X_{0,0} = 
    \begin{pmatrix}
        0 & 1 & 0 & 1 & 0 & 0 & 0 & 0 & 0 \\
        0 & 0 & 0 & 0 & 2 & 0 & 0 & 0 & 2 \\
        0 & 0 & 0 & 0 & 0 & 2 & 0 & 2 & 0 \\
        2 & 0 & 0 & 0 & -2 & 0 & 0 & 0 & -4 \\
        0 & 0 & 2 & 0 & 0 & 0 & 0 & 0 & 0 \\
        0 & -2 & 0 & 2 & 0 & 0 & 0 & 0 & 0 \\
        0 & 0 & 0 & 0 & 0 & 0 & 2 & 0 & 0 \\
        0 & 0 & 0 & 0 & 0 & -2 & 0 & 2 & 0 \\
        2 & 0 & 0 & 0 & -2 & 0 & 0 & 0 & 2
    \end{pmatrix}.
\end{equation}

The fusion tensors $Y_{a,b}$ are submatrices of $X_{a,b}$. Specifically, $Y_{1,1}$ is the first three rows of $X_{1,1}$, $Y_{0,1}$ is the first two rows of $X_{0,1}$, $Y_{1,0}$ is the first two rows of $X_{1,0}$, and $Y_{0,0}$ is the first three rows of $X_{0,0}$. The right-inverse $Y_{a,b}^{-1}$ is formed by taking the corresponding columns of $(X_{a,b})^{-1}$. 

\section{Explicit form of multiplication of $\mathcal{A},\mA^*$ and a basis transformation}
\label{app:explicit-mul-basis-trans}

Let's denote the basis of $\mA$ as
\begin{equation}
    \basis=\{e_0^{12},e_0^{13},e_0^{22},e_0^{23},e_1^{11},e_1^{12},e_1^{21},e_1^{22}\},
\end{equation}
and relabel the basis elements as
\begin{equation}
    \basis=\{e_1,e_2,e_3,e_4,e_5,e_6,e_7,e_8\}.
\end{equation}
The explicit form of multiplication of $\mA$ can be computed via eq.~\eqref{eqn:multiplication-lambda} using the explicit expressions of $Y_{a,b}$ and calculating $Y_{a,b}(e_a^{mn} \otimes e_b^{pq})Y_{a,b}^{-1}$. The result is 
\begin{equation}
\label{eq:mul-epl}
\begin{aligned}
    &\basis^T \basis = \\
    &\begin{pmatrix}
        e_3 & e_4 & e_1 & e_2 & -e_5 & -e_6 & e_7 & e_8\\
        e_4 & e_3 & e_2 & e_1 & e_6 & e_5 & -e_8 & -e_7\\
        e_1 & e_2 & e_3 & e_4 & e_5 & e_6 & e_7 & e_8\\
        e_2 & e_1 & e_4 & e_3 & -e_6 & -e_5 & -e_8 & -e_7\\
        e_5 & e_6 & e_5 & e_6 & 0 & 0 & \frac{e_4-e_2}{2} & \frac{e_3-e_1}{2}\\
        e_6 & e_5 & e_6 & e_5 & 0 & 0 & \frac{e_1-e_3}{2} & \frac{e_2-e_4}{2}\\
        -e_7 & -e_8 & e_7 & e_8 & \frac{-e_2-e_4}{2} & \frac{-e_1-e_3}{2} & 0 & 0\\
        -e_8 & -e_7 & e_8 & e_7 & \frac{e_1+e_3}{2} & \frac{e_2+e_4}{2} & 0 & 0
    \end{pmatrix}
\end{aligned}
\end{equation}
One can see the unit of $\mA$ is $1=e_3$. 

A more straightforward way to obtain the multiplication structure is to take the special case $N=2$ of eq.~\eqref{eq:mpo_mult}, where
\begin{equation}
\begin{aligned}
    &O^{(2)}(e_1)=-Z\otimes \bo,\quad O^{(2)}(e_2)=Z\otimes Z\\
    &O^{(2)}(e_3)=\bo\otimes\bo,\quad O^{(2)}(e_4)=-\bo\otimes Z\\
    &O^{(2)}(e_5)=\text{diag}[1,-1,0,0] X\otimes X\\
    &O^{(2)}(e_6)=\text{diag}[1,1,0,0] X\otimes X\\
    &O^{(2)}(e_7)=\text{diag}[0,0,-1,-1] X\otimes X\\
    &O^{(2)}(e_8)=\text{diag}[0,0,-1,1] X\otimes X
\end{aligned}
\end{equation}
are distinct operators. Namely, $O^{(2)}$ is a faithful representation of $\mA$. 
From this, one can compute $O^{(2)}(e_I) O^{(2)}(e_J)=\sum_K \lambda_{IJ}^K O^{(2)}(e_K)$ and obtain eq.~\eqref{eq:mul-epl}. 

To see that the algebra $\mA$ is isomorphic to the direct sum of matrix algebras $\mathcal{M}_2\oplus \mathcal{M}_2$, we use a basis transformation $R$ that brings basis $\{e_I\}$ to $\{f_I\}$, explicitly,
\begin{equation}
    R = \frac{1}{4}\begin{pmatrix}
        -1 & 1 & 1 & -1 & 0 & 0 & 0 & 0\\
        0 & 0 & 0 & 0 & 2 & 2 & 0 & 0\\
        0 & 0 & 0 & 0 & 0 & 0 & -2 & 2\\
        1 & 1 & 1 & 1 & 0 & 0 & 0 &0 \\
        -1 & -1 & 1 & 1 & 0 & 0 & 0 & 0\\
        0 & 0 & 0 & 0 & -2 & 2 & 0 & 0\\
        0 & 0 & 0 & 0 & 0 & 0 & -2 & -2\\
        1 & -1 & 1 & -1 & 0 & 0 & 0 & 0
    \end{pmatrix},
\end{equation}
such that $f_I = \sum_J R_{IJ} e_J$. The multiplication structure in the basis $\{f_I\}$ is related to that in the basis $\{e_I\}$ via
\begin{equation}
    (\lambda_f)_{IJ}^K= \sum_{I' J' K'} R_{II'} R_{JJ'} (R^{-1})_{K' K}(\lambda_e)_{I' J'}^{K'},
\end{equation}
and $\lambda_f$ is the multiplication structure constant of $\mathcal{M}_2\oplus \mathcal{M}_2$. A more straightforward way to see $\mA\cong \mathcal{M}_2\oplus \mathcal{M}_2$ is by observing the explicit form of $O^{(2)}(f_I)$. 

The multiplication of $\mA^*$ is given in eq.~\eqref{eqn:mul-As}, which comes from the comultiplication of $\mA$. In the basis $\basis=\{e_I\}$, the comultiplication structure is simple by using $O^{(l_1+l_2)}(e_a^{mn})=O^{(l_1)}\otimes O^{(l_2)}\circ\Delta(e_a^{mn})$ for any $l_1,l_2$. Let $A^{(l)}_a$ denote the $l$-fold blocking of the tensor $A_a$ in the horizontal direction, then $O^{(l)}(e_a^{mn})=(A^{(l)}_a)_{mn}$ and therefore
\begin{equation}
\begin{aligned}
    O^{(l_1+l_2)}(e_a^{mn})&=(A^{(l_1+l_2)}_a)_{mn}\\&=\sum_p (A^{(l_1)}_a)_{mp}\otimes (A^{(l_2)}_a)_{pn}\\
    &=\sum_p  O^{(l_1)}(e_a^{mp})\otimes  O^{(l_2)}(e_a^{pn}),
\end{aligned}
\end{equation}
which is eq.~\eqref{eqn:A-comultiplication} in the main text and leads to eq.~\eqref{eqn:mul-As}. 
Let us denote the dual basis of $\basis=\{e_I\}$ as $\basis^*=\{e^I\}$, the explicit form of multiplication of $\mA^*$ is then
\begin{equation}
    (\basis^*)^T \basis^*=\begin{pmatrix}
        0 & 0 & e^1 & e^2 & 0 & 0 & 0 & 0\\
        0 & 0 & 0 & 0 & 0 & 0 & 0 & 0 \\
        0 & 0 & e^3 & e^4 & 0 & 0 & 0 & 0\\
        0 & 0 & 0 & 0 & 0 & 0 & 0 & 0 \\
        0 & 0 & 0 & 0 & e^5 & e^6 & 0 & 0 \\
        0 & 0 & 0 & 0 & 0 & 0 & e^5 & e^6 \\
        0 & 0 & 0 & 0 & e^7 & e^8 & 0 & 0 \\
        0 & 0 & 0 & 0 & 0 & 0 & e^7 & e^8
    \end{pmatrix}.
\end{equation}
One can see that $\mA^*$ does not have a unit because $e^2\cdot x=0$ for any $x\in \mA^*$. Note that in the basis $\{f_I\}$, the comultiplication structure is more complicated, and the explicit form of $\Lambda_f$ can be obtained from $\Lambda_e$ via
\begin{equation}
    (\Lambda_f)^{IJ}_K=\sum_{I'J'K'} (R^{-1})_{I' I} (R^{-1})_{J' J} R_{KK'} (\Lambda_e)^{I'J'}_{K'}.
\end{equation} 

The operator $*$ in $\mA$ is 
\begin{equation}
\begin{aligned}
    &e_I^* = e_I,\quad I=1,2,3,4\\
    &e_5^* = e_8,\quad e_6^*=-e_7,
\end{aligned}
\end{equation}
which satisfies all axioms of the $*$ operator, and $O^{(2)}$ is a faithful $*$-representation.  

\section{Representations of $\mA,\mA^*$}
\label{app:rep}
\subsection{Representations of $\mA$}
As discussed in the main text, the basis elements of $\mA$ are chosen such that $O^{(1)}(e_a^{mn})=(A_a)_{mn}$. It turns out that $O^{(1)}$ is an irreducible representation which we identify as $\phi_\mathbf{1}$, and the explicit form is,
\begin{equation}
    \begin{aligned}
        &\phi_\mathbf{1}(e_1)=\phi_\mathbf{1}(e_4)=-Z\\
        &\phi_\mathbf{1}(e_2)=\phi_\mathbf{1}(e_3)=\bo\\
        &\phi_\mathbf{1}(e_5)=\phi_\mathbf{1}(e_6)=\begin{pmatrix}
            0 & 1\\
            0 & 0
        \end{pmatrix}\\
        &\phi_\mathbf{1}(e_8)=-\phi_\mathbf{1}(e_7)=\begin{pmatrix}
            0 & 0\\
            1 & 0
        \end{pmatrix}.
    \end{aligned}
\end{equation}
Since $\mA\cong\mathcal{M}_2\oplus\mathcal{M}_2$, there are two 2d irreducible representations $\phi_{\vbnt}$, with $\vbnt=\mathbf{1},\mathbf{2}$. $\phi_\mathbf{1}$ is identified using $\phi_\mathbf{1}(e_a^{mn}):=O^{(1)}(e_a^{mn})$ above, and the explicit form of $\phi_\mathbf{2}$ is
\begin{equation}
\begin{aligned}
    &\phi_\mathbf{2}(e_4)=-\phi_\mathbf{2}(e_1)=Z\\
    &\phi_\mathbf{2}(e_3)=-\phi_\mathbf{2}(e_2)=\bo\\
&\phi_\mathbf{2}(e_6)=-\phi_\mathbf{2}(e_5)=\begin{pmatrix}
            0 & 1\\
            0 & 0
        \end{pmatrix}\\
&\phi_\mathbf{2}(e_7)=\phi_\mathbf{2}(e_8)=\begin{pmatrix}
            0 & 0\\
            -1 & 0
        \end{pmatrix}.
\end{aligned}
\end{equation}

\subsection{Representations of $\mA^*$}

In this section, we provide details of the representations of $\mA^*$ as a non-unital and non-semisimple algebra. We say that a module is \textit{indecomposable} if it is non-zero and cannot be written as a direct sum of two non-zero submodules. Being a non-semisimple algebra, $\mA^*$ has reducible but indecomposable modules. 

To study the representation of $\mA^*$, we first unitize $\mA^*$ to $\mA^*_+:=\mA^*\oplus\mathbb{C}$ by adding an identity element $e^0=(0,1)$ where $0\in\mA^*$, with the multiplication rule,
\begin{equation}
    e^0\cdot e^I = e^I \cdot e^0 = e^I,\quad\forall I
\end{equation}
where we redefine as $e^I\equiv(e^I,0)$. 
As an algebra, $\mA^*_+$ is unital but still non-semisimple, and with 9 basis elements. We can study the representations of $\mA^*$ by studying that of $\mA^*_+$ due to the following proposition,

\begin{prop}
    Let $\mA$ be a non-unital associative finite-dimensional algebra and $\mA_+$ its unitization. The representations of $\mA$ are in bijection with that of $\mA_+$. In particular, the irreducible representations of $\mA$ are in bijection with the irreducible representations of $\mA_+$.
\end{prop}

For semisimple algebra, all irreducible representations (simple modules) can be obtained by decomposing its left regular representation (due to Artin-Wedderburn theorem). Similarly, one can still obtain all irreducible representations of $\mA^*_+$ from its left regular representation, stated as follows~\cite{assem2006elements}.

\begin{prop}
Let $\mA$ be a finite-dimensional unital associatve algebra over $\mathbb{C}$ and $m_{\mathrm{reg}}$ be its left regular module. Then,
\begin{equation}
m_{\mathrm{reg}}\cong\bigoplus_{a\in \mathrm{Irr}(\mA)} P_a^{\oplus n_a}
\end{equation}
where $P_a$ is an projective indecomposable module such that $P_a/\mathrm{rad}(P_a)=S_a$ where $S_a$ is the simple module labeled by $a$; $P_a^{\oplus n_a}$ means the direct sum of $n_a$ copies of $P_a$, and $\cong$ denotes up to a similarity transformation.
\end{prop}

A module is \textit{projective} if it is a direct summand of a free module. Since the left regular module is a free module, by definition $P_a$'s in the above decomposition are projective modules. 
Given a simple module $S$, its projective cover is unique.

In practice, projective indecomposable modules can be found by identifying primitive idempotents using the following proposition,

\begin{prop}
Let $\mA$ be a finite-dimensional unital associatve algebra over field $\mathbb{C}$ and $E\in \mA$ an idempotent. Then $\mA E$ is an indecomposable left $\mA$-module if and only if $E$ is a primitive idempotent.
\end{prop}

Since $\mA^*_+$ is unital, the unit $e^0$ is always an idempotent and $\mA^*_+ e^0$ is the regular module $m_\reg$. By decomposing $e^0$ into summation of primitive idempotents $e^0=\sum_a E_a$, one can obtain all indecomposable modules that are submodules of $m_\reg$.

We now apply the above general results to our algebra $\mA^*_+$. Specifically, the primitive idempotents are 
\begin{equation}
\begin{aligned}
   & E_1 = e^3 = e^0_{22} \\
   & E_2 = e^5 = e^1_{11} \\ 
   & E_3 = e^8 = e^1_{22} \\
   & E_0 = e^0 - (E_1+E_2+E_3).
\end{aligned}
\end{equation}
Each $E_a$ defines a vector space $V_a:=(\mA^*_+) E_a$. Explicitly, $V_1=\mathrm{span}(e^1,e^3)$, $V_2=\mathrm{span}(e^5,e^7)$, $V_3=\mathrm{span}(e^6,e^8)$, $V_0=\mathrm{span}(e^2,e^4,e^0-e^3-e^5-e^8)$.
The left regular module is then decomposed into a direct sum of projective indecomposable modules
\begin{equation}
   m_\reg \cong P_1 \oplus P_2 \oplus P_3 \oplus P_0
\end{equation}
with properties summarized in Table.~\ref{tbl:modules}, and $P_2\cong P_3$. Explicitly, the corresponding irreducible representation $\psi_{P_a}$ are
\begin{equation}
\begin{aligned}
    &\psi_{P_1}(e^1)=\begin{pmatrix}
        0 & 1\\
        0 & 0
    \end{pmatrix},\quad 
    \psi_{P_1}(e^3)=\begin{pmatrix}
        0 & 0\\
        0 & 1
    \end{pmatrix}\\
    &\psi_{P_1}(e^0)=\bo_2,
\end{aligned}
\end{equation}
and
\begin{equation}
\begin{aligned}
    &\psi_{P_2}(e^5) = \begin{pmatrix}
        1 & 0\\
        0 & 0
    \end{pmatrix},\quad \psi_{P_2}(e^6) = \begin{pmatrix}
        0 & 1\\
        0 & 0
    \end{pmatrix}\\
    &\psi_{P_2}(e^7) = \begin{pmatrix}
        0 & 0\\
        1 & 0
    \end{pmatrix},\quad \psi_{P_2}(e^8) = \begin{pmatrix}
        0 & 0\\
        0 & 1
    \end{pmatrix},\\
    &\psi_{P_2}(e^0)=\bo_2,
\end{aligned}
\end{equation}
and $\psi_{P_3}=\psi_{P_2}$, and
\begin{equation}
\begin{aligned}
    &\psi_{P_0}(e^1) = \begin{pmatrix}
       0& 1 & 0\\
       0& 0 & 0\\
       0 & 0 & 0
    \end{pmatrix},\quad \psi_{P_0}(e^2) = \begin{pmatrix}
        0 &0 & 1\\
        0 & 0& 0\\
        0 & 0 & 0
    \end{pmatrix}\\
    &\psi_{P_0}(e^3) = \begin{pmatrix}
        0 & 0 & 0\\
        0 & 1 & 0\\
        0 & 0 & 0
    \end{pmatrix},\quad \psi_{P_0}(e^4) = \begin{pmatrix}
       0 & 0 & 0\\
       0& 0 & 1\\
       0 & 0 & 0
    \end{pmatrix},\\
    &\psi_{P_0}(e^0)=\bo_3.
\end{aligned}
\end{equation}

To obtain the simple modules, we need to compute the quotient $S_a=P_a/\mathrm{rad}(P_a)$. 
The radical of the module $P_0$ is $\text{span}(e^2,e^4)$, and the radical of the module $P_1$ is $\text{span}(e^1)$. 
After quotienting out the radical, we obtain three simple modules $S_0,S_1,S_2$. Explicitly, the corresponding irreducible representations $\psi_{S_a}$ are as follows. $S_0$ corresponds to a trivial 1d representation that maps everything to 0 except $e^0$,
\begin{equation}
    \psi_{S_0}(e^0) = 1.
\end{equation}
$S_1$ is a nontrivial 1d representation,
\begin{equation}
    \psi_{S_1}(e^0) = 1,\quad \psi_{S_1}(e^3) = 1,
\end{equation}
that together with the representation $\phi_\mathbf{1}$ of $\mA$, constructs an MPO tensor by eq.~\eqref{eqn:Ma-def} and this MPO tensor generates the many-body identity matrix. $\psi_{S_2}=\psi_{P_2}$ is a 2d representation, that together with the representation $\phi_\mathbf{1}$ of $\mA$, constructs an MPO tensor by eq.~\eqref{eqn:Ma-def} and this MPO tensor generates $U_{CZY}$. 

\section{Introduction to $C^*$-weak Hopf algebras}
\label{app:intro-algebra}


In the literature, the theory for MPO algebras as well as the construction of MPDO renormalization fixed points are based on weak Hopf algebras or their representation categories. 
In this appendix, we introduce the notions of pre-bialgebra and $C^*$-weak Hopf algebra~\cite{bohm_coassociativec_1996,bohm_weak_1999,bohm_weak_2000}, and then present the associated construction of RFP tensors. 

\begin{defn}[Pre-bialgebra]
    A pre-bialgebra $\mA$ is an associative algebra together with a linear map $\Delta:\mA\ra \mA\times \mA$ called coproduct, which is associative
    \begin{equation}
    \label{eqn:coproduct-associative}
        (\Delta\otimes \Id)\circ\Delta = (\Id\otimes \Delta)\circ\Delta,
    \end{equation}
    and the coproduct $\Delta$ is multiplicative, i.e., for all $x,y\in\mA$,
    \begin{equation}
        \Delta(xy)=\Delta(x)\Delta(y),
    \end{equation}
    where the multiplication on $\mA\otimes\mA$ is defined component-wise, $(x\otimes y)(z\otimes w)=(xz\otimes yw)$ for $x,y,z,w\in \mathcal{A}$.

    The pre-bialgebra is called unital if $\mA$ is unital; and co-unital if there is a linear functional $1_{\mA^*}: \mA\ra \mathbb{C}$ called counit such that
    \begin{equation}
        (\cu\otimes\Id)\circ\Delta=(\Id\otimes\cu)\circ\Delta=\Id.
    \end{equation}
\end{defn}

We use the terminologies ``product'' and ``multiplication'', ``coproduct'' and ``comultiplication'' interchangeably. 

Given a pre-bialgebra $\mA$ with multiplication $\cdot_\mA$ and coproduct $\Delta_\mA$, its dual vector space $\mA^*=\text{Hom}(\mA,\mathbb{C})$ naturally inherits a pre-bialgebra structure, with multiplication $\cdot_{\mA^*}=\Delta_\mA^T$ and $\Delta_{\mA^*}=(\cdot_\mA)^T$. The unit of $\mA^*$ is the counit of $\mA$, if it exists. Choosing a basis $\basis=\{e_I\}$ of $\mA$, one defines the dual basis $\basis^*=\{e^I\}$ of $\mA^*$ such that $e^J(e_I)=\delta^J_I$. The product and coproduct on $\mA$ can be specify using the algebra basis,
\begin{equation}
    e_I e_J=\sum_K \lambda_{IJ}^K e_K,\quad \Delta(e_I)=\sum_{JK} \Lambda_I^{JK} e_J \otimes e_K,
\end{equation}
which induce the coproduct and product on $\mA^*$,
\begin{equation}
    \Delta(e^I)=\sum_{JK} \lambda^I_{JK} e^J\otimes e^K,\quad e^I e^J=\sum_K \Lambda^{IJ}_K e^K.
\end{equation}
The above formulas lead to
\begin{equation}
\label{eqn:concatenate-two}
    \sum_{IJ} e_I\otimes e_J \otimes e^I e^J = \sum_K \Delta(e_K) \otimes e^K.
\end{equation}
This relation will be useful in the concatenation of MPO tensors. 

\begin{defn}[$C^*$-pre-bialgebra]
    A $C^*$-pre-bialgebra $\mA$ is a unital pre-bialgebra endowed with an anti-linear map $*:\mA\ra \mA$ which is an involution $x^{**}=x$, anti-homomorphism $(xy)^*=y^* x^*$, and cohomomorphism 
    \begin{equation}
    \Delta\circ *=(*\otimes *)\circ\Delta;    
    \end{equation}
    and $\mA$ has a faithful $*$-representation $\phi$, i.e., a faithful representation $\phi:\mA\ra \text{End}(V)$ such that  $\phi(x^*)=\phi(x)^\dg$ for all $x\in\mA$ where $\dg$ denotes the conjugate transpose on $\text{End}(V)$. 
\end{defn}

Note that a finite-dimensional $C^*$-pre-bialgebra as a $C^*$-algebra must be unital because it is a direct sum of matrix algebras. 
A $C^*$-weak Hopf algebra is a counital $C^*$-pre-bialgebra with extra structures, formally defined as follows,

\begin{defn}[$C^*$-weak Hopf algebra]
A $C^*$-weak Hopf algebra $\mA$ is a counital $C^*$-pre-bialgebra such that the counit $\cu\in \mA^*$ is weakly comultiplicative
\[
\begin{aligned}
 \cu(xyz)&=\cu(x y_{(1)})\cu(y_{(2)}z)\\
 &=\cu(x y_{(2)}) \cu(y_{(1)} z),   
\end{aligned}
\]
for all $x,y,z\in\mA$; the unit $1\in\mA$ is weakly comultiplicative
\[
\begin{aligned}
1_{(1)}\otimes 1_{(2)}\otimes 1_{(3)}&=1_{(1)}\otimes 1_{(2)} 1_{(1')}\otimes 1_{(2')}\\
&=1_{(1)}\otimes 1_{(1')} 1_{(2)} \otimes 1_{(2')};    
\end{aligned}
\]
and $\mA$ is endowed with a linear map $S:\mA\ra\mA$  called antipode, which is anti-multiplicative and anti-comultiplicative, and satisfies
\[
\begin{aligned}
    & S(x_{(1)}) x_{(2)}=\cu(1_{(1)}x)1_{(2)}\\
    &x_{(1)}S(x_{(2)})=1_{(1)}\cu(x 1_{(2)}).
\end{aligned}
\]
\end{defn}

The dual $\mA^*$ of a $C^*$-WHA $\mA$ is also a $C^*$-WHA, with the $*$-operation on $\mA^*$ induced from that of $\mA$, and possesses a faithful $*$-representation $\psi$. 
By Tannaka duality, the $*$-representation category $\text{Rep}^*(\mA)$ of a $C^*$-weak Hopf algebra $\mA$ is a unitary multifusion category. 
The existence of counit ensures the existence of a monoidal unit called the trivial representation. A $C^*$-weak Hopf algebra is \textit{biconnected} if both the trivial representation of $\mA$ and $\mA^*$ are irreducible. Let $\mA$ be a biconnected $C^*$-weak Hopf algebra, there exists a special element called ``canonical regular element'' $\omega\in\mA^*$, that takes the form
\begin{equation}
    \omega=\sum_{\vbnt\in\text{Irr}(\mA)} \frac{d_\vbnt}{\mathrm{FPdim}(\mathcal D)} \textbf{x}_\vbnt,\quad \text{with} \quad\textbf{x}_\vbnt=\tr\circ \phi_\vbnt,
\end{equation}
where $\text{Irr}(\mA)$ is the set of all irreducible representations of the algebra $\mA$, $\phi_\vbnt$ is the irreducible representation of $\mA$ labeled by $\vbnt$, $d_\vbnt$ is the quantum dimension of $\phi_\vbnt$, and $\mathrm{FPdim}(\mathcal D):=\sum_\vbnt d_\vbnt^2$. Given a faitful $*$-representation $\phi$, the weight matrix $b(\omega)$ is defined such that $\tr(b(\omega)\phi(x))=\omega(x)$ for all $x\in\mA$. 

We are now ready to state the RFP construction from a $C^*$-weak Hopf algebra~\cite{molnar2022matrix,ruiz2024matrix},

\begin{thm}
Given a biconnected $C^*$-weak Hopf algebra $\mA$, and $\phi$ and $\psi$ being faithful $*$-representations of $\mA$ and $\mA^*$ respectively, the MPDO generated by tensor $M$
\begin{equation}
\label{eqn:WHA-FP}
\begin{array}{c}
        \begin{tikzpicture}[scale=1.,baseline={([yshift=-0.65ex] current bounding box.center)}]
		\draw (-0.75,0) node {$\alpha$};
		\draw (0.75,0) node {$\beta$};
		\draw (0,0.75) node {$i$};
        \draw (0,-0.75) node {$j$};
        \whaM{(0,0)}{2};
        \end{tikzpicture}
        \end{array}=\sum_{I\in\basis} [b(\omega) \phi(e_I)]_{ij} [\psi(e^I)]_{\alpha\beta},
\end{equation}  
and boundary condition matrix $B(x)$ defined by
\begin{equation}
    \tr[B(x)\psi(f)]=f(x),\quad \forall f\in\mA^*
\end{equation}
is an RFP MPDO, for all positive nonzero $x\in\mA$~\cite{ruiz2024matrix}. 
\end{thm}

The proof of this theorem relies on the structure of a biconnected $C^*$-weak Hopf algebra. We will present a different proof in app.~\ref{app:proof-suff}, from which the algebraic structure can be relaxed to a $C^*$-pre-bialgebra that is not necessarily counital and not necessarily co-semisimple (meaning, $\mA^*$ is not unital and is not semisimple), along with certain conditions. 

\section{Proof of Theorem~\ref{prop:sufficient-main}}
\label{app:proof-suff}

In this appendix, we prove Theorem~\ref{prop:sufficient-main}, which is a sufficient condition for a general construction of an RFP, with input $\mA$ being an associative semisimple $C^*$-pre-bialgebra, possibly lacking a counit and not necessarily cosemisimple (i.e., $\mA^*$ is not necessarily unital and semisimple). 

\subsection{Structure of MPDO renormalization fixed points}
The proof will utilize the structure of the MPDO tensor, in particular, its vertical canonical form. 

\begin{prop}[Vertical canonical form]
For any tensor $A$ generating an MPDO, it is always possible to obtain another tensor $M$ that generates the same MPDO for any $N$, and $M$ after a basis transformation by isometry $U$ is in the \textit{vertical canonical form},
   \begin{equation}
   \label{eqn:vcf-main}
    U M_{(\alpha\beta)} U^\dg=\bigoplus_\vbnt \mu_\vbnt\otimes M_{(\alpha\beta),\vbnt}
\end{equation} 
where $\mu_\vbnt$ are diagonal and positive matrices, and $\{M_\vbnt\}$ form a \textit{basis of normal tensors} (BNT): (i) Each $M_\vbnt$ is a normal tensor, in the sense that the algebra generated by the set $\{M_{(\alpha\beta),\vbnt}\}_{\alpha\beta}$ is a full matrix algebra, and $M_\vbnt$ is normalized; (ii) For $\vbnt\neq \vbnt'$, the normal tensors $M_\vbnt$ and $M_{\vbnt'}$ are independent, in the sense that $M_{\vbnt'}$ cannot be brought to $M_\vbnt$ by $M_{(\alpha\beta)\vbnt'}=e^{i\phi} XM_{(\alpha\beta)\vbnt}X^{-1}$ where $X$ is an invertible matrix and $\phi$ is a phase~\cite{cirac2017matrix}. 
\end{prop}

All operations in eq.~\eqref{eqn:vcf-main} are in the vertical direction. The tensor $M_\vbnt$ is normalized such that the transfer matrix $\sum_{\alpha\beta}M_{(\alpha\beta)\vbnt}\otimes \bar{M}_{(\alpha\beta)\vbnt}$ has spectral radius 1 where $\bar{M}_\vbnt$ denotes the complex conjugate of $M_\vbnt$. The vertical canonical form plays a central role in verifying RFP condition (eq.~\eqref{eqn:RFP-def-main}), as stated below. 
\begin{thm}
\label{thm:fusion-main}
    A tensor $M$ generating an MPDO is an RFP iff there exists isometries $W_{\mathbf{ab}}$ (satisfying $W_{\mathbf{ab}}(W_{\mathbf{ab}})^\dg = \bo$), such that 
    \begin{equation}
\begin{array}{c}
        \begin{tikzpicture}[scale=1,baseline={([yshift=-0.75ex] current bounding box.center)}
        ]
        \def\rd{0.247};
        \def\rlen{0.318};
        \def\wth{0.25};
        \whaMorg{(0,0)}{2}{$M_\vbnt$}{whampdocolor};
        \whaMorg{(\whadx,0)}{2}{$M_{\mathbf{b}}$}{whampdocolor};
        \vtTensor{(\whadx*0.5, {\wth+\rd+\rlen})}{\whadx*0.7}{\wth}{\whadx*0.5}{{\wth+\rlen}}{\small $W_{\mathbf{ab}}$};
         \vTensor{(\whadx*0.5, -0.815)}{\whadx*0.7}{0.25}{\whadx*0.5}{0.568}{\small $W^\dg_{\mathbf{ab}}$};
        \end{tikzpicture}
        \end{array}=\bigoplus_{\mathbf{c}} \chi_{\mathbf{a,b,c}} \otimes
\begin{array}{c}
        \begin{tikzpicture}[scale=1.,baseline={([yshift=-0.75ex] current bounding box.center)}
        ]
       \whaMorg{(0,0)}{2}{$M_{\mathbf{c}}$}{whampdocolor};
        \end{tikzpicture}
        \end{array}
\end{equation}
where each $\chi_{\mathbf{a,b,c}}$ is a diagonal matrix with positive diagonal elements, and $m_{\mathbf{c}}=\sum_{\mathbf{ab}} \tr[\chi_{\mathbf{a,b,c}}] m_\vbnt m_\mathbf{b}$ where $m_\vbnt:=\tr[\mu_\vbnt]$.  
\end{thm}

The construction of the fixed-point tensor $M$ with input $\mA$ amounts to constructing tensors $\{M_\vbnt\}$ that verify Theorem~\ref{thm:fusion-main}. We will choose $M_\vbnt$ being proportional to $\tilde{M}_\vbnt$; recall that the explicit form of $\tilde{M}_\vbnt$ is
\begin{equation}
    \label{eqn:explicit-Mabar}(\tilde{M}_\vbnt)^{ij}_{\alpha\beta}=\sum_{I\in\basis}[\phi_\vbnt(e_I)]_{ij} [\psi(e^I)]_{\alpha\beta}.
\end{equation}

\subsection{Proof}

To begin, we note that associativity of the coproduct eq.~\eqref{eqn:coproduct-associative} imposes that the fusion multiplicities $N_{\mathbf{ab}}^\mathbf{c}$ are \textit{associative}, i.e.,
\begin{equation}
    \sum_{\mathbf{e}\in\text{Irr}(\mA)} N_{\mathbf{ab}}^\mathbf{e} N_{\mathbf{ec}}^\mathbf{d} = \sum_{\mathbf{f}\in\text{Irr}(\mA)}N_{\mathbf{af}}^\mathbf{d} N_{\mathbf{bc}}^\mathbf{f} .
\end{equation}
Equivalently, define the left multiplication matrices $N_\vbnt$ with components $(N_\vbnt)_{\mathbf{cb}}=N_{\mathbf{ab}}^\mathbf{c}$. 
Then, the associative condition reads
\begin{equation}
\label{eqn:N-algebra}
    N_\vbnt N_\mathbf{b} = \sum_{\mathbf{c}\in\text{Irr}(\mA)} N^\mathbf{c}_{\mathbf{ab}} N_\mathbf{c}.
\end{equation}
One can also define the right multiplication matrices $\tilde{N}_\vbnt$ via $(\tilde{N}_\mathbf{b})_{\mathbf{ca}}=N_{\mathbf{ab}}^\mathbf{c}$, and the associative condition implies the commutation relation
\begin{equation}
    \tilde{N}_\mathbf{c} N_\mathbf{a} = N_\mathbf{a} \tilde{N}_\mathbf{c}.
\end{equation}
We say that $N_{\mathbf{bc}}^\mathbf{a}$ are \textit{transitive} if for any $\mathbf{a,b}\in\text{Irr}(\mA)$, there exsits $\mathbf{c,d}\in\text{Irr}(\mA)$ such that $N_{\mathbf{ac}}^\mathbf{d}>0, N_{\mathbf{da}}^\mathbf{b}>0$.

Before presenting the proof, we point out that the construction is essentially the same as in~\cite{ruiz2024matrix}. The nontrivial aspect is the relaxation of structural assumptions on $\mA$, as the proof in~\cite{ruiz2024matrix} relies on $\mA$ being a biconnected $C^*$-weak Hopf algebra. We note that $\mA$ being a biconnected $C^*$-weak Hopf algebra is sufficient to fulfill (1) and (2) in Theorem~\ref{prop:sufficient-main}, but not necessary.

\begin{proof}[Proof of Theorem~\ref{prop:sufficient-main}]

Consider the tensors $\{\tilde{M}_\vbnt\}$ in eq.~\eqref{eqn:explicit-Mabar} and denote the normalized $\tilde{M}_\vbnt$ as $M_\vbnt$ with $M_\vbnt=C_\vbnt \tilde{M}_\vbnt$, where the constant $C_\vbnt$ is fixed by requiring $M_\vbnt$ to be normalized. By construction, $M_\vbnt$ is a (vertical) normal tensor because $\psi$ is a faithful representation of $\mA^*$ and the density theorem of $\phi_\vbnt$ states that $\text{alg}(\{\phi_\vbnt(e_I)\}_I)=\mathcal{M}_{D_\vbnt}$ where $D_\vbnt$ is the dimension of $\phi_\vbnt$. The tensors formed by different irreps $\vbnt$ are independent normal tensors; therefore, $\{M_\vbnt\}$ form a basis of normal tensors. 

Concatenating the tensors $\tilde{M}_\vbnt$ and $\tilde{M}_{\mathbf{b}}$ and using eq.~\eqref{eqn:concatenate-two}, 
\begin{equation}
\label{eqn:Mab}
\begin{aligned}
    \begin{array}{c}
        \begin{tikzpicture}[scale=1,baseline={([yshift=-0.75ex] current bounding box.center)}
        ]
        \def\rd{0.247};
        \def\rlen{0.318};
        \def\wth{0.25};
        \whaMorg{(0,0)}{2}{$\tilde{M}_\vbnt$}{white};
        \whaMorg{(\whadx,0)}{2}{$\tilde{M}_\mathbf{b}$}{white};
        \end{tikzpicture}
        \end{array}&=\sum_{IJ}\phi_\vbnt(e_I)\otimes\phi_{\mathbf{b}}(e_J)\otimes \psi(e^I e^J)\\
        &=\sum_I \left[(\phi_\vbnt\otimes\phi_{\mathbf{b}})\circ\Delta(e_I)\right]\otimes\psi(e^I)\\
        &\simeq \sum_I \left[\bigoplus_{\mathbf{c}} \bo_{N_{\mathbf{ab}}^{\mathbf{c}}}\otimes \phi_{\mathbf{c}}(e_I)\right]\otimes \psi(e^I)\\
        &=\bigoplus_{\mathbf{c}} \bo_{N_{\mathbf{ab}}^\mathbf{c}} \otimes
\begin{array}{c}
        \begin{tikzpicture}[scale=1.,baseline={([yshift=-0.75ex] current bounding box.center)}
        ]
       \whaMorg{(0,0)}{2}{$\tilde{M}_\mathbf{c}$}{white};
        \end{tikzpicture}
        \end{array},
\end{aligned}
\end{equation}
where $\simeq$ means that equality holds up to an isometry. This isometry is identified with $W_{\mathbf{ab}}$ (the cohomomorphism of $*$ operator ($\Delta\circ *=(*\otimes *)\circ\Delta$) guarantees that $W_{\mathbf{ab}}$ is an isometry).
We obtain that $\chi_{\mathbf{a,b,c}}$ is proportional to the identity matrix 
\begin{equation}
    \chi_{\mathbf{a,b,c}}=\frac{C_\vbnt C_\mathbf{b}}{C_\mathbf{c}} \bo_{N_{\mathbf{ab}}^\mathbf{c}}.
\end{equation}
In the remaining, we show that choosing a suitable superposition coefficient of $\tilde{M}_\vbnt$ as $d_\vbnt/\mathrm{FPdim}(\mathcal D)$ verifies $m_{\mathbf{c}}=\sum_{\mathbf{ab}} \tr[\chi_{\mathbf{a,b,c}}] m_\vbnt m_\mathbf{b}$ where $m_\vbnt:=\tr[\mu_\vbnt]$. Under this choice,  $m_\vbnt = d_\vbnt/(\mathrm{FPdim}(\mathcal D) C_\vbnt)$. 

The proof follows~\cite{etingof2015tensor}. 
    Define a matrix $\tilde{N}_S=\sum_{\vbnt\in\text{Irr}(\mA)} \tilde{N}_\vbnt$, which has strictly positive entries due to transitivity of $N_{\mathbf{ab}}^\mathbf{c}$. The Perron-Frobenius theorem states that $\tilde{N}_S$ has a unique eigenvalue equal to its spectral radius, and the corresponding eigenvector $\textbf{v}$ can be normalized to have strictly positive entries. Since $[\tilde{N}_S, N_\vbnt]=0$ for any $\vbnt$, $\textbf{v}$ is also an eigenvector of $N_\vbnt$ with eigenvalues $d_\vbnt$. 

    
    Since $N_\vbnt$ as a matrix with non-negative entries has an eigenvector $\textbf{v}$ with strictly positive entries, by the Perron-Frobenius theorem, the corresponding eigenvalue $d_\vbnt$ is the spectral radius. Acting eq.~\eqref{eqn:N-algebra} on $\textbf{v}$ then leads to
    \begin{equation}
        d_\vbnt d_\mathbf{b} = \sum_\mathbf{c} N^\mathbf{c}_{\mathbf{ab}} d_\mathbf{c} ,
    \end{equation}
    which can be rewritten as $\sum_\mathbf{c} d_\mathbf{c} (N_\vbnt)_{\mathbf{cb}}=d_\vbnt d_\mathbf{b}$. 
    Therefore,
    \begin{equation}
    \begin{aligned}
        \sum_{\mathbf{ab}} N^\mathbf{c}_{\mathbf{ab}} d_\mathbf{a} d_\mathbf{b} &=  \sum_{\mathbf{ab}} d_\mathbf{b} (N_\mathbf{a})_{\mathbf{cb}} d_\mathbf{a}  \\
        & = \sum_{\mathbf{a}^* \mathbf{b}} d_\mathbf{b} (N_{\mathbf{a}^*})_{\mathbf{bc}} d_{\mathbf{a}^*} = \mathrm{FPdim}(\mathcal D) d_\mathbf{c}, 
    \end{aligned}
    \end{equation}
    where $\mathbf{a}^*$ denotes the irreducible representation satisfying $N_{\mathbf{a}^*}=N_\mathbf{a}^T$, and $d_\mathbf{a}=d_{\mathbf{a}^*}$ follows from the fact that $N_\mathbf{a}$ and $N_\mathbf{a}^T$ have the same spectral radius. 
    
    Recall that $m_\vbnt = d_\vbnt/(\mathrm{FPdim}(\mathcal D) C_\vbnt)$, then $\{m_\vbnt\}$ fulfill 
    \begin{equation}
    \begin{aligned}
        m_\mathbf{c} &= \sum_{\mathbf{ab}} \left(\frac{C_\mathbf{a} C_\mathbf{b}}{C_\mathbf{c}} N_{\mathbf{ab}}^\mathbf{c} \right)  m_\mathbf{a} m_\mathbf{b}\\
        &=\sum_{\mathbf{ab}} \tr[\chi_{\mathbf{a,b,c}}] m_\vbnt m_\mathbf{b}.
    \end{aligned}
    \end{equation}
    This finishes the proof. 
\end{proof}

Still, one needs to prove that after taking the periodic boundary condition, $\rho^{(N)}(M)$ is a valid density matrix. Note that the tensor $M$ in eq.~\eqref{eqn:fixed-point-generate-main} generates the following density matrix,
\begin{equation}
    \rho^{(N)}(M)=\sum_I b^{\otimes N}\cdot\phi^{\otimes N}\circ \Delta^{N-1}(e_I) \tr[\psi(e^I)],
\end{equation}
where $\phi=\oplus_{\vbnt\in\mathrm{Irr}(\mA)}\phi_\vbnt$ and matrix $b=\oplus_{\vbnt\in\mathrm{Irr}(\mA)}d_\vbnt \bo_{D_\vbnt}/\mathrm{FPdim}(\mathcal D)$ is a diagonal matrix with $D_\vbnt$ being the dimension of irreducible representation $\phi_\vbnt$. 
We identify $x=\sum_I c_I e_I$ where $c_I=\tr[\psi(e^I)]$, from which one can rewrite $\rho^{(N)}(M)= b^{\otimes N} \cdot\phi^{\otimes N}\circ \Delta^{N-1}(x)$. If there exists $y\in\mA$ such that $x=y y^*$, then $\rho^{(N)}(M)$ is a valid density matrix because
\begin{equation}
\begin{aligned}
    &\rho^{(N)}(M)= \\
    &\left[\sqrt{b}^{\otimes N}\phi^{\otimes N}\circ \Delta^{N-1}(y)\right]\left[\sqrt{b}^{\otimes N}\phi^{\otimes N}\circ \Delta^{N-1}(y)\right]^\dg.
\end{aligned}
\end{equation}
Furthermore, such $\rho^{(N)}$ is a locally purified density operator. 

In the example of CZY, the choice of $\psi=\psi_{S_1}\oplus \psi_{S_2}$ leads to
\[x=e_3+e_5+e_8\]
which is indeed positive since $x=yy^*$ where $y=-\frac{1}{2}e_1+\frac{1}{2}e_3 - e_8$. 


\section{Generalizations to finite groups}
\label{sec:finite group generalization}

In this section we show how to generalize the pre-algebra construction presented above to finite groups with a non-trivial 3-cocycle. Similar to eq.~\eqref{eq:CZX to semion}, these pre-bialgebras are related to the boundaries of twisted G-injective PEPS.

We start from the the generalization of the CZY MPO to a finite group $G$ with a non-trivial 3-cocycle $\omega$ \cite{garrerubio2023classifying}: Consider the MPO tensors $A_g$, $g\in G$ defined by the non-zero components
\begin{equation}
    A_g^{gh,h} = \omega_g\ket{h}\bra{h},
\end{equation}
where $\omega_g = \sum_{kl} \omega(g,k,k^{-1}l) \ket{k}\bra{l}$. The PBC MPOs defined by these tensors form a representation of the group $G$. This representation of the group is anomalous, and the anomaly is exactly given by the equivalence class of the three-cocycle $\omega$. Moreover, one can explicitly verify that,  if $\omega_g$ is invertible, $A_g$ is injective after blocking two sites. In the example of $G=\mathbb{Z}_2=\{0,1\}$, $\omega_0$ is not invertible while $\omega_1$ is invertible; therefore $A_0$ is non-injective while $A_1$ is injective, agreeing with the discussion in the main text.  

Just as in the case of the CZY model, we define the corresponding algebra in the injective sectors as the set of MPOs on two sites with arbitrary boundary condition:
\begin{equation}\label{eq:G algebra def wrong}
    \begin{aligned}
    &\mathcal{A}_g = \\&\mathrm{Span}\left\{\sum_{l} \bra{k}\omega_g \ket{l} \bra{l} \omega_g \ket{h} \cdot \ket{gl,gh}\bra{l,h} \middle| h,k\in G \right\}.        
    \end{aligned}
\end{equation}
As $\omega_g$ is invertible, a basis of this set is 
\begin{equation}\label{eq:G algebra def right}
    \mathcal{B}_g = \left\{ b_g^{k,h} =\bra{k} \omega_g \ket{h} \cdot \ket{gk,gh}\bra{k,h} \middle| h,k\in G \right\}.
\end{equation}
One can readily verify that setting $\mathcal{A}_g = \mathrm{Span}\  \mathcal{B}_g$ in the non-injective sectors as well, the set $\mA = \bigoplus_g \mathcal{A}_g$ is closed under multiplication, i.e., it is an algebra. Explicitly, as $\bra{k}\omega_g\ket{l} \neq 0$, the set
\begin{equation}\label{eq:G algebra def right 2}
    \mathcal{B}'_g = \left\{  \ket{gk,gh}\bra{k,h} \middle| h,k\in G \right\}
\end{equation}
also forms a basis of $\mathcal{A}_g$: The basis elements $\bigcup_g \mathcal{B}'_g$ themselves form a closed set under multiplication. 

We now define a coproduct on this algebra corresponding to growing the MPO. Explicitly, in the basis $\bigcup_g\mathcal{B}_g$, the coproduct is defined as
\begin{equation}
    \Delta(b_g^{k,h}) = \sum_{lm} \omega(g,l,l^{-1}m) \cdot b_g^{k,l} \otimes b_g^{m,h}.
\end{equation}
One can readily verify that $\Delta$ is associative, i.e., it is a coproduct, and that it is multiplicative, making $\mathcal{A}$ a pre-bialgebra. It is easy to check that in this pre-bialgebra there is no counit. The comultiplication of the unit, however, is very nice: If $\omega$ is a normalized 3-cocycle, i.e., $\omega(1,k,l) = 1$ for all $k,l\in G$, then $\Delta(1) = 1\otimes 1$. In the example of $G=\mathbb{Z}_2$, the pre-bialgebra structure of $\mA$ is the same as in the main text, up to a basis transformation from $\{e_I\}$ to $\{b_I\}$. 

Let us note that this construction, in the non-injective sectors, removes the non-invertible matrix $\omega_g$ for the last MPO tensor, and thus these MPOs do not fit into the MPO algebraic framework presented in \cite{molnar2022matrix}. One can, however, obtain a different MPO description (that is different from the previous MPOs only in the non-injective sectors) by starting from the pre-bialgebra defined above and applying the procedure outlined in sec.~\ref{sec:mpo-reconstruction}. Carrying out this procedure for the case $G=\mathbb{Z}_2$, we obtain the MPOs of the CZY model described in the main text.

Let us finally compare this pre-bialgebra with the weak Hopf algebra obtained from the boundary of the twisted $G$-injective PEPS. The injective MPO tensors describing the boundary of the twisted $G$-injective PEPS are given by
\begin{equation*}
    \begin{aligned}
    T_g = &\sum_{h_1,h_2}  \omega(g,h_1, h_1^{-1}h_2) \cdot \\
    & \cdot \ket{gh_2, gh_1} \bra{h_2,h_1} \otimes \ket{gh_2, h_2}\bra{gh_1,h_1} ,  
    \end{aligned}
\end{equation*}
where the first component in the tensor product is the physical index of the MPO tensor, while the second component is its virtual index. As each MPO tensor $T_g$ is injective, we do not need further blocking to obtain the MPO algebra. Closing $T_g$ with arbitrary boundary conditions gives rise to the same set $\mA_g$ as in the construction above: the boundary $\omega(g,h,^{-1}k) \ket{gh,h}\bra{gk,k}$ generates the algebra element $\ket{gk,gh}\bra{k,h}$. 

Therefore the boundary of the twisted $G$-injective PEPS gives rise to the same algebraic structure $\mA$ as the generalized CZY model (meaning, the multiplication rules are the same). The MPO tensors $T_g$ and $A_g$, however, are different, and thus the coalgebra structure defined by $T_g$ is different from the coalgebra structure defined by the tensors $A_g$.  Explicitly, the coproduct $\hat{\Delta}$ defined by the tensor $T_g$ is 
\begin{equation*}
    \hat{\Delta}(b_g^{k,h}) = \sum_{l} b_g^{k,l} \otimes b_g^{l,h}.
\end{equation*}
It is straightforward to check that $\hat{\Delta}(1)\neq 1\otimes 1$ and that $(\mathcal{A},\hat{\Delta})$ is a weak Hopf algebra: the counit and the antipode are given by
\begin{align}
      \epsilon(b^g_{k,l}) &= \delta_{k,l}, \\
      S(b^g_{k,l}) &= \frac{\omega(g^{-1},g,l)}{\omega(g^{-1},g,k)\omega(1,l,l^{-1}k)}b^{g^{-1}}_{gl,gk}.    
\end{align}
In the example of $G=\mathbb{Z}_2$, this is $\mA_\bdy$ in sec.~\ref{sec:relation}.

\bibliography{main}

\begin{thebibliography}{43}%
\makeatletter
\providecommand \@ifxundefined [1]{%
 \@ifx{#1\undefined}
}%
\providecommand \@ifnum [1]{%
 \ifnum #1\expandafter \@firstoftwo
 \else \expandafter \@secondoftwo
 \fi
}%
\providecommand \@ifx [1]{%
 \ifx #1\expandafter \@firstoftwo
 \else \expandafter \@secondoftwo
 \fi
}%
\providecommand \natexlab [1]{#1}%
\providecommand \enquote  [1]{``#1''}%
\providecommand \bibnamefont  [1]{#1}%
\providecommand \bibfnamefont [1]{#1}%
\providecommand \citenamefont [1]{#1}%
\providecommand \href@noop [0]{\@secondoftwo}%
\providecommand \href [0]{\begingroup \@sanitize@url \@href}%
\providecommand \@href[1]{\@@startlink{#1}\@@href}%
\providecommand \@@href[1]{\endgroup#1\@@endlink}%
\providecommand \@sanitize@url [0]{\catcode `\\12\catcode `\$12\catcode
  `\&12\catcode `\#12\catcode `\^12\catcode `\_12\catcode `\%12\relax}%
\providecommand \@@startlink[1]{}%
\providecommand \@@endlink[0]{}%
\providecommand \url  [0]{\begingroup\@sanitize@url \@url }%
\providecommand \@url [1]{\endgroup\@href {#1}{\urlprefix }}%
\providecommand \urlprefix  [0]{URL }%
\providecommand \Eprint [0]{\href }%
\providecommand \doibase [0]{https://doi.org/}%
\providecommand \selectlanguage [0]{\@gobble}%
\providecommand \bibinfo  [0]{\@secondoftwo}%
\providecommand \bibfield  [0]{\@secondoftwo}%
\providecommand \translation [1]{[#1]}%
\providecommand \BibitemOpen [0]{}%
\providecommand \bibitemStop [0]{}%
\providecommand \bibitemNoStop [0]{.\EOS\space}%
\providecommand \EOS [0]{\spacefactor3000\relax}%
\providecommand \BibitemShut  [1]{\csname bibitem#1\endcsname}%
\let\auto@bib@innerbib\@empty
\bibitem [{\citenamefont {Seifnashri}\ and\ \citenamefont
  {Shirley}(2025)}]{Seifnashri:2025vhf}%
  \BibitemOpen
  \bibfield  {author} {\bibinfo {author} {\bibfnamefont {S.}~\bibnamefont
  {Seifnashri}}\ and\ \bibinfo {author} {\bibfnamefont {W.}~\bibnamefont
  {Shirley}},\ }\href@noop {} {\bibinfo {title} {{Disentangling anomaly-free
  symmetries of quantum spin chains}}} (\bibinfo {year} {2025}),\ \Eprint
  {https://arxiv.org/abs/2503.09717} {arXiv:2503.09717 [cond-mat.str-el]}
  \BibitemShut {NoStop}%
\bibitem [{\citenamefont {Tu}\ \emph {et~al.}(2025)\citenamefont {Tu},
  \citenamefont {Long},\ and\ \citenamefont {Else}}]{Tu:2025bqf}%
  \BibitemOpen
  \bibfield  {author} {\bibinfo {author} {\bibfnamefont {Y.-T.}\ \bibnamefont
  {Tu}}, \bibinfo {author} {\bibfnamefont {D.~M.}\ \bibnamefont {Long}},\ and\
  \bibinfo {author} {\bibfnamefont {D.~V.}\ \bibnamefont {Else}},\ }\href@noop
  {} {\bibinfo {title} {{Anomalies of global symmetries on the lattice}}}
  (\bibinfo {year} {2025}),\ \Eprint {https://arxiv.org/abs/2507.21209}
  {arXiv:2507.21209 [cond-mat.str-el]} \BibitemShut {NoStop}%
\bibitem [{\citenamefont {Shirley}\ \emph {et~al.}(2025)\citenamefont
  {Shirley}, \citenamefont {Zhang}, \citenamefont {Ji},\ and\ \citenamefont
  {Levin}}]{Shirley:2025yji}%
  \BibitemOpen
  \bibfield  {author} {\bibinfo {author} {\bibfnamefont {W.}~\bibnamefont
  {Shirley}}, \bibinfo {author} {\bibfnamefont {C.}~\bibnamefont {Zhang}},
  \bibinfo {author} {\bibfnamefont {W.}~\bibnamefont {Ji}},\ and\ \bibinfo
  {author} {\bibfnamefont {M.}~\bibnamefont {Levin}},\ }\href@noop {} {\bibinfo
  {title} {{Anomaly-free symmetries with obstructions to gauging and
  onsiteability}}} (\bibinfo {year} {2025}),\ \Eprint
  {https://arxiv.org/abs/2507.21267} {arXiv:2507.21267 [cond-mat.str-el]}
  \BibitemShut {NoStop}%
\bibitem [{\citenamefont {Nielsen}\ and\ \citenamefont
  {Ninomiya}(1981)}]{nielsen1981no}%
  \BibitemOpen
  \bibfield  {author} {\bibinfo {author} {\bibfnamefont {H.~B.}\ \bibnamefont
  {Nielsen}}\ and\ \bibinfo {author} {\bibfnamefont {M.}~\bibnamefont
  {Ninomiya}},\ }\bibfield  {title} {\bibinfo {title} {A no-go theorem for
  regularizing chiral fermions},\ }\href
  {https://doi.org/https://doi.org/10.1016/0370-2693(81)91026-1} {\bibfield
  {journal} {\bibinfo  {journal} {Physics Letters B}\ }\textbf {\bibinfo
  {volume} {105}},\ \bibinfo {pages} {219} (\bibinfo {year}
  {1981})}\BibitemShut {NoStop}%
\bibitem [{\citenamefont {Friedan}(1982)}]{friedan1982proof}%
  \BibitemOpen
  \bibfield  {author} {\bibinfo {author} {\bibfnamefont {D.}~\bibnamefont
  {Friedan}},\ }\bibfield  {title} {\bibinfo {title} {A proof of the
  nielsen-ninomiya theorem},\ }\href
  {https://doi.org/https://doi.org/10.1007/BF01403500} {\bibfield  {journal}
  {\bibinfo  {journal} {Communications in Mathematical Physics}\ }\textbf
  {\bibinfo {volume} {85}},\ \bibinfo {pages} {481} (\bibinfo {year}
  {1982})}\BibitemShut {NoStop}%
\bibitem [{\citenamefont {Verstraete}\ \emph {et~al.}(2004)\citenamefont
  {Verstraete}, \citenamefont {Garc\'{\i}a-Ripoll},\ and\ \citenamefont
  {Cirac}}]{verstraete2004matrix}%
  \BibitemOpen
  \bibfield  {author} {\bibinfo {author} {\bibfnamefont {F.}~\bibnamefont
  {Verstraete}}, \bibinfo {author} {\bibfnamefont {J.~J.}\ \bibnamefont
  {Garc\'{\i}a-Ripoll}},\ and\ \bibinfo {author} {\bibfnamefont {J.~I.}\
  \bibnamefont {Cirac}},\ }\bibfield  {title} {\bibinfo {title} {Matrix product
  density operators: Simulation of finite-temperature and dissipative
  systems},\ }\href {https://doi.org/10.1103/PhysRevLett.93.207204} {\bibfield
  {journal} {\bibinfo  {journal} {Phys. Rev. Lett.}\ }\textbf {\bibinfo
  {volume} {93}},\ \bibinfo {pages} {207204} (\bibinfo {year}
  {2004})}\BibitemShut {NoStop}%
\bibitem [{\citenamefont {Pirvu}\ \emph {et~al.}(2010)\citenamefont {Pirvu},
  \citenamefont {Murg}, \citenamefont {Cirac},\ and\ \citenamefont
  {Verstraete}}]{pirvu2010matrix}%
  \BibitemOpen
  \bibfield  {author} {\bibinfo {author} {\bibfnamefont {B.}~\bibnamefont
  {Pirvu}}, \bibinfo {author} {\bibfnamefont {V.}~\bibnamefont {Murg}},
  \bibinfo {author} {\bibfnamefont {J.~I.}\ \bibnamefont {Cirac}},\ and\
  \bibinfo {author} {\bibfnamefont {F.}~\bibnamefont {Verstraete}},\ }\bibfield
   {title} {\bibinfo {title} {Matrix product operator representations},\ }\href
  {https://doi.org/https://doi.org/10.1088/1367-2630/12/2/025012} {\bibfield
  {journal} {\bibinfo  {journal} {New Journal of Physics}\ }\textbf {\bibinfo
  {volume} {12}},\ \bibinfo {pages} {025012} (\bibinfo {year}
  {2010})}\BibitemShut {NoStop}%
\bibitem [{\citenamefont {Haegeman}\ and\ \citenamefont
  {Verstraete}(2017)}]{haegeman2017diagonalizing}%
  \BibitemOpen
  \bibfield  {author} {\bibinfo {author} {\bibfnamefont {J.}~\bibnamefont
  {Haegeman}}\ and\ \bibinfo {author} {\bibfnamefont {F.}~\bibnamefont
  {Verstraete}},\ }\bibfield  {title} {\bibinfo {title} {Diagonalizing transfer
  matrices and matrix product operators: A medley of exact and computational
  methods},\ }\href
  {https://doi.org/https://doi.org/10.1146/annurev-conmatphys-031016-025507}
  {\bibfield  {journal} {\bibinfo  {journal} {Annual Review of Condensed Matter
  Physics}\ }\textbf {\bibinfo {volume} {8}},\ \bibinfo {pages} {355} (\bibinfo
  {year} {2017})}\BibitemShut {NoStop}%
\bibitem [{\citenamefont {Buerschaper}(2014)}]{buerschaper2014twisted}%
  \BibitemOpen
  \bibfield  {author} {\bibinfo {author} {\bibfnamefont {O.}~\bibnamefont
  {Buerschaper}},\ }\bibfield  {title} {\bibinfo {title} {Twisted injectivity
  in projected entangled pair states and the classification of quantum
  phases},\ }\href {https://doi.org/https://doi.org/10.1016/j.aop.2014.09.007}
  {\bibfield  {journal} {\bibinfo  {journal} {Annals of Physics}\ }\textbf
  {\bibinfo {volume} {351}},\ \bibinfo {pages} {447} (\bibinfo {year}
  {2014})}\BibitemShut {NoStop}%
\bibitem [{\citenamefont {{\c{S}}ahino{\u{g}}lu}\ \emph
  {et~al.}(2021)\citenamefont {{\c{S}}ahino{\u{g}}lu}, \citenamefont
  {Williamson}, \citenamefont {Bultinck}, \citenamefont {Mari{\"e}n},
  \citenamefont {Haegeman}, \citenamefont {Schuch},\ and\ \citenamefont
  {Verstraete}}]{csahinouglu2021characterizing}%
  \BibitemOpen
  \bibfield  {author} {\bibinfo {author} {\bibfnamefont {M.~B.}\ \bibnamefont
  {{\c{S}}ahino{\u{g}}lu}}, \bibinfo {author} {\bibfnamefont {D.}~\bibnamefont
  {Williamson}}, \bibinfo {author} {\bibfnamefont {N.}~\bibnamefont
  {Bultinck}}, \bibinfo {author} {\bibfnamefont {M.}~\bibnamefont
  {Mari{\"e}n}}, \bibinfo {author} {\bibfnamefont {J.}~\bibnamefont
  {Haegeman}}, \bibinfo {author} {\bibfnamefont {N.}~\bibnamefont {Schuch}},\
  and\ \bibinfo {author} {\bibfnamefont {F.}~\bibnamefont {Verstraete}},\
  }\bibfield  {title} {\bibinfo {title} {Characterizing topological order with
  matrix product operators},\ }in\ \href
  {https://doi.org/https://doi.org/10.1007/s00023-020-00992-4} {\emph {\bibinfo
  {booktitle} {Annales Henri Poincar{\'e}}}},\ Vol.~\bibinfo {volume} {22}\
  (\bibinfo {organization} {Springer},\ \bibinfo {year} {2021})\ pp.\ \bibinfo
  {pages} {563--592}\BibitemShut {NoStop}%
\bibitem [{\citenamefont {Bultinck}\ \emph {et~al.}(2017)\citenamefont
  {Bultinck}, \citenamefont {Mari{\"e}n}, \citenamefont {Williamson},
  \citenamefont {{\c{S}}ahino{\u{g}}lu}, \citenamefont {Haegeman},\ and\
  \citenamefont {Verstraete}}]{bultinck2017anyons}%
  \BibitemOpen
  \bibfield  {author} {\bibinfo {author} {\bibfnamefont {N.}~\bibnamefont
  {Bultinck}}, \bibinfo {author} {\bibfnamefont {M.}~\bibnamefont
  {Mari{\"e}n}}, \bibinfo {author} {\bibfnamefont {D.~J.}\ \bibnamefont
  {Williamson}}, \bibinfo {author} {\bibfnamefont {M.~B.}\ \bibnamefont
  {{\c{S}}ahino{\u{g}}lu}}, \bibinfo {author} {\bibfnamefont {J.}~\bibnamefont
  {Haegeman}},\ and\ \bibinfo {author} {\bibfnamefont {F.}~\bibnamefont
  {Verstraete}},\ }\bibfield  {title} {\bibinfo {title} {Anyons and matrix
  product operator algebras},\ }\href
  {https://doi.org/https://doi.org/10.1016/j.aop.2017.01.004} {\bibfield
  {journal} {\bibinfo  {journal} {Annals of physics}\ }\textbf {\bibinfo
  {volume} {378}},\ \bibinfo {pages} {183} (\bibinfo {year}
  {2017})}\BibitemShut {NoStop}%
\bibitem [{\citenamefont {Williamson}\ \emph {et~al.}(2016)\citenamefont
  {Williamson}, \citenamefont {Bultinck}, \citenamefont {Mari{\"e}n},
  \citenamefont {{\c{S}}ahino{\u{g}}lu}, \citenamefont {Haegeman},\ and\
  \citenamefont {Verstraete}}]{williamson2016matrix}%
  \BibitemOpen
  \bibfield  {author} {\bibinfo {author} {\bibfnamefont {D.~J.}\ \bibnamefont
  {Williamson}}, \bibinfo {author} {\bibfnamefont {N.}~\bibnamefont
  {Bultinck}}, \bibinfo {author} {\bibfnamefont {M.}~\bibnamefont
  {Mari{\"e}n}}, \bibinfo {author} {\bibfnamefont {M.~B.}\ \bibnamefont
  {{\c{S}}ahino{\u{g}}lu}}, \bibinfo {author} {\bibfnamefont {J.}~\bibnamefont
  {Haegeman}},\ and\ \bibinfo {author} {\bibfnamefont {F.}~\bibnamefont
  {Verstraete}},\ }\bibfield  {title} {\bibinfo {title} {Matrix product
  operators for symmetry-protected topological phases: Gauging and edge
  theories},\ }\href
  {https://doi.org/https://doi.org/10.1103/PhysRevB.94.205150} {\bibfield
  {journal} {\bibinfo  {journal} {Physical Review B}\ }\textbf {\bibinfo
  {volume} {94}},\ \bibinfo {pages} {205150} (\bibinfo {year}
  {2016})}\BibitemShut {NoStop}%
\bibitem [{\citenamefont {Lootens}\ \emph {et~al.}(2021)\citenamefont
  {Lootens}, \citenamefont {Fuchs}, \citenamefont {Haegeman}, \citenamefont
  {Schweigert},\ and\ \citenamefont {Verstraete}}]{lootens2021matrix}%
  \BibitemOpen
  \bibfield  {author} {\bibinfo {author} {\bibfnamefont {L.}~\bibnamefont
  {Lootens}}, \bibinfo {author} {\bibfnamefont {J.}~\bibnamefont {Fuchs}},
  \bibinfo {author} {\bibfnamefont {J.}~\bibnamefont {Haegeman}}, \bibinfo
  {author} {\bibfnamefont {C.}~\bibnamefont {Schweigert}},\ and\ \bibinfo
  {author} {\bibfnamefont {F.}~\bibnamefont {Verstraete}},\ }\bibfield  {title}
  {\bibinfo {title} {{Matrix product operator symmetries and intertwiners in
  string-nets with domain walls}},\ }\href
  {https://doi.org/10.21468/SciPostPhys.10.3.053} {\bibfield  {journal}
  {\bibinfo  {journal} {SciPost Phys.}\ }\textbf {\bibinfo {volume} {10}},\
  \bibinfo {pages} {053} (\bibinfo {year} {2021})}\BibitemShut {NoStop}%
\bibitem [{\citenamefont {Moln{\'a}r}\ \emph {et~al.}(2022)\citenamefont
  {Moln{\'a}r}, \citenamefont {{Ruiz-de-Alarc{\'o}n}}, \citenamefont
  {Garre-Rubio}, \citenamefont {Schuch}, \citenamefont {Cirac},\ and\
  \citenamefont {P{\'e}rez-Garc{\'i}a}}]{molnar2022matrix}%
  \BibitemOpen
  \bibfield  {author} {\bibinfo {author} {\bibfnamefont {A.}~\bibnamefont
  {Moln{\'a}r}}, \bibinfo {author} {\bibfnamefont {A.}~\bibnamefont
  {{Ruiz-de-Alarc{\'o}n}}}, \bibinfo {author} {\bibfnamefont {J.}~\bibnamefont
  {Garre-Rubio}}, \bibinfo {author} {\bibfnamefont {N.}~\bibnamefont {Schuch}},
  \bibinfo {author} {\bibfnamefont {J.~I.}\ \bibnamefont {Cirac}},\ and\
  \bibinfo {author} {\bibfnamefont {D.}~\bibnamefont {P{\'e}rez-Garc{\'i}a}},\
  }\bibfield  {title} {\bibinfo {title} {Matrix product operator algebras {I}:
  representations of weak {Hopf} algebras and projected entangled pair
  states},\ }\href {https://arxiv.org/abs/2204.05940} {\bibfield  {journal}
  {\bibinfo  {journal} {arXiv:2204.05940}\ } (\bibinfo {year}
  {2022})}\BibitemShut {NoStop}%
\bibitem [{\citenamefont {Bridgeman}\ \emph {et~al.}(2023)\citenamefont
  {Bridgeman}, \citenamefont {Lootens},\ and\ \citenamefont
  {Verstraete}}]{bridgeman2023invertible}%
  \BibitemOpen
  \bibfield  {author} {\bibinfo {author} {\bibfnamefont {J.~C.}\ \bibnamefont
  {Bridgeman}}, \bibinfo {author} {\bibfnamefont {L.}~\bibnamefont {Lootens}},\
  and\ \bibinfo {author} {\bibfnamefont {F.}~\bibnamefont {Verstraete}},\
  }\bibfield  {title} {\bibinfo {title} {Invertible bimodule categories and
  generalized schur orthogonality},\ }\href
  {https://doi.org/https://doi.org/10.1007/s00220-023-04781-y} {\bibfield
  {journal} {\bibinfo  {journal} {Communications in Mathematical Physics}\
  }\textbf {\bibinfo {volume} {402}},\ \bibinfo {pages} {2691} (\bibinfo {year}
  {2023})}\BibitemShut {NoStop}%
\bibitem [{\citenamefont {Lootens}\ \emph {et~al.}(2023)\citenamefont
  {Lootens}, \citenamefont {Delcamp}, \citenamefont {Ortiz},\ and\
  \citenamefont {Verstraete}}]{lootens2023dualities}%
  \BibitemOpen
  \bibfield  {author} {\bibinfo {author} {\bibfnamefont {L.}~\bibnamefont
  {Lootens}}, \bibinfo {author} {\bibfnamefont {C.}~\bibnamefont {Delcamp}},
  \bibinfo {author} {\bibfnamefont {G.}~\bibnamefont {Ortiz}},\ and\ \bibinfo
  {author} {\bibfnamefont {F.}~\bibnamefont {Verstraete}},\ }\bibfield  {title}
  {\bibinfo {title} {Dualities in one-dimensional quantum lattice models:
  Symmetric hamiltonians and matrix product operator intertwiners},\ }\href
  {https://doi.org/https://doi.org/10.1103/PRXQuantum.4.020357} {\bibfield
  {journal} {\bibinfo  {journal} {PRX Quantum}\ }\textbf {\bibinfo {volume}
  {4}},\ \bibinfo {pages} {020357} (\bibinfo {year} {2023})}\BibitemShut
  {NoStop}%
\bibitem [{\citenamefont {Lootens}\ \emph {et~al.}(2025)\citenamefont
  {Lootens}, \citenamefont {Delcamp},\ and\ \citenamefont
  {Verstraete}}]{lootens2025entanglement}%
  \BibitemOpen
  \bibfield  {author} {\bibinfo {author} {\bibfnamefont {L.}~\bibnamefont
  {Lootens}}, \bibinfo {author} {\bibfnamefont {C.}~\bibnamefont {Delcamp}},\
  and\ \bibinfo {author} {\bibfnamefont {F.}~\bibnamefont {Verstraete}},\
  }\bibfield  {title} {\bibinfo {title} {Entanglement and the density matrix
  renormalization group in the generalized landau paradigm},\ }\href
  {https://doi.org/https://doi.org/10.1038/s41567-025-02961-2} {\bibfield
  {journal} {\bibinfo  {journal} {Nature Physics}\ ,\ \bibinfo {pages} {1}}
  (\bibinfo {year} {2025})}\BibitemShut {NoStop}%
\bibitem [{\citenamefont {Levin}\ and\ \citenamefont
  {Gu}(2012)}]{levin2012braiding}%
  \BibitemOpen
  \bibfield  {author} {\bibinfo {author} {\bibfnamefont {M.}~\bibnamefont
  {Levin}}\ and\ \bibinfo {author} {\bibfnamefont {Z.-C.}\ \bibnamefont {Gu}},\
  }\bibfield  {title} {\bibinfo {title} {Braiding statistics approach to
  symmetry-protected topological phases},\ }\href
  {https://doi.org/https://doi.org/10.1103/PhysRevB.86.115109} {\bibfield
  {journal} {\bibinfo  {journal} {Physical Review B}\ }\textbf {\bibinfo
  {volume} {86}},\ \bibinfo {pages} {115109} (\bibinfo {year}
  {2012})}\BibitemShut {NoStop}%
\bibitem [{\citenamefont {Chen}\ \emph {et~al.}(2011)\citenamefont {Chen},
  \citenamefont {Liu},\ and\ \citenamefont {Wen}}]{chen2011two}%
  \BibitemOpen
  \bibfield  {author} {\bibinfo {author} {\bibfnamefont {X.}~\bibnamefont
  {Chen}}, \bibinfo {author} {\bibfnamefont {Z.-X.}\ \bibnamefont {Liu}},\ and\
  \bibinfo {author} {\bibfnamefont {X.-G.}\ \bibnamefont {Wen}},\ }\bibfield
  {title} {\bibinfo {title} {Two-dimensional symmetry-protected topological
  orders and their protected gapless edge excitations},\ }\href
  {https://doi.org/10.1103/PhysRevB.84.235141} {\bibfield  {journal} {\bibinfo
  {journal} {Phys. Rev. B}\ }\textbf {\bibinfo {volume} {84}},\ \bibinfo
  {pages} {235141} (\bibinfo {year} {2011})}\BibitemShut {NoStop}%
\bibitem [{\citenamefont {Chatterjee}\ \emph {et~al.}(2025)\citenamefont
  {Chatterjee}, \citenamefont {Pace},\ and\ \citenamefont
  {Shao}}]{chatterjee2025quantized}%
  \BibitemOpen
  \bibfield  {author} {\bibinfo {author} {\bibfnamefont {A.}~\bibnamefont
  {Chatterjee}}, \bibinfo {author} {\bibfnamefont {S.~D.}\ \bibnamefont
  {Pace}},\ and\ \bibinfo {author} {\bibfnamefont {S.-H.}\ \bibnamefont
  {Shao}},\ }\bibfield  {title} {\bibinfo {title} {Quantized axial charge of
  staggered fermions and the chiral anomaly},\ }\href
  {https://doi.org/https://doi.org/10.1103/PhysRevLett.134.021601} {\bibfield
  {journal} {\bibinfo  {journal} {Physical Review Letters}\ }\textbf {\bibinfo
  {volume} {134}},\ \bibinfo {pages} {021601} (\bibinfo {year}
  {2025})}\BibitemShut {NoStop}%
\bibitem [{\citenamefont {Pace}\ \emph {et~al.}(2025)\citenamefont {Pace},
  \citenamefont {Chatterjee},\ and\ \citenamefont {Shao}}]{pace2025lattice}%
  \BibitemOpen
  \bibfield  {author} {\bibinfo {author} {\bibfnamefont {S.~D.}\ \bibnamefont
  {Pace}}, \bibinfo {author} {\bibfnamefont {A.}~\bibnamefont {Chatterjee}},\
  and\ \bibinfo {author} {\bibfnamefont {S.-H.}\ \bibnamefont {Shao}},\
  }\bibfield  {title} {\bibinfo {title} {Lattice t-duality from non-invertible
  symmetries in quantum spin chains},\ }\href
  {https://doi.org/https://doi.org/10.21468/SciPostPhys.18.4.121} {\bibfield
  {journal} {\bibinfo  {journal} {SciPost Physics}\ }\textbf {\bibinfo {volume}
  {18}},\ \bibinfo {pages} {121} (\bibinfo {year} {2025})}\BibitemShut
  {NoStop}%
\bibitem [{\citenamefont {Kock}(2008)}]{kock2008elementary}%
  \BibitemOpen
  \bibfield  {author} {\bibinfo {author} {\bibfnamefont {J.}~\bibnamefont
  {Kock}},\ }\bibfield  {title} {\bibinfo {title} {Elementary remarks on units
  in monoidal categories},\ }in\ \href@noop {} {\emph {\bibinfo {booktitle}
  {Mathematical Proceedings of the Cambridge Philosophical Society}}},\ Vol.\
  \bibinfo {volume} {144}\ (\bibinfo {organization} {Cambridge University
  Press},\ \bibinfo {year} {2008})\ pp.\ \bibinfo {pages} {53--76}\BibitemShut
  {NoStop}%
\bibitem [{\citenamefont {Etingof}\ \emph {et~al.}(2015)\citenamefont
  {Etingof}, \citenamefont {Gelaki}, \citenamefont {Nikshych},\ and\
  \citenamefont {Ostrik}}]{etingof2015tensor}%
  \BibitemOpen
  \bibfield  {author} {\bibinfo {author} {\bibfnamefont {P.}~\bibnamefont
  {Etingof}}, \bibinfo {author} {\bibfnamefont {S.}~\bibnamefont {Gelaki}},
  \bibinfo {author} {\bibfnamefont {D.}~\bibnamefont {Nikshych}},\ and\
  \bibinfo {author} {\bibfnamefont {V.}~\bibnamefont {Ostrik}},\ }\href@noop {}
  {\emph {\bibinfo {title} {Tensor categories}}},\ Vol.\ \bibinfo {volume}
  {205}\ (\bibinfo  {publisher} {American Mathematical Soc.},\ \bibinfo {year}
  {2015})\BibitemShut {NoStop}%
\bibitem [{\citenamefont {Zwolak}\ and\ \citenamefont
  {Vidal}(2004)}]{zwolak2004mixed}%
  \BibitemOpen
  \bibfield  {author} {\bibinfo {author} {\bibfnamefont {M.}~\bibnamefont
  {Zwolak}}\ and\ \bibinfo {author} {\bibfnamefont {G.}~\bibnamefont {Vidal}},\
  }\bibfield  {title} {\bibinfo {title} {Mixed-state dynamics in
  one-dimensional quantum lattice systems: A time-dependent superoperator
  renormalization algorithm},\ }\href
  {https://doi.org/10.1103/PhysRevLett.93.207205} {\bibfield  {journal}
  {\bibinfo  {journal} {Phys. Rev. Lett.}\ }\textbf {\bibinfo {volume} {93}},\
  \bibinfo {pages} {207205} (\bibinfo {year} {2004})}\BibitemShut {NoStop}%
\bibitem [{\citenamefont {{Cirac}}\ \emph {et~al.}(2017)\citenamefont
  {{Cirac}}, \citenamefont {{P{\'e}rez-Garc{\'\i}a}}, \citenamefont
  {{Schuch}},\ and\ \citenamefont {{Verstraete}}}]{cirac2017matrix}%
  \BibitemOpen
  \bibfield  {author} {\bibinfo {author} {\bibfnamefont {J.~I.}\ \bibnamefont
  {{Cirac}}}, \bibinfo {author} {\bibfnamefont {D.}~\bibnamefont
  {{P{\'e}rez-Garc{\'\i}a}}}, \bibinfo {author} {\bibfnamefont
  {N.}~\bibnamefont {{Schuch}}},\ and\ \bibinfo {author} {\bibfnamefont
  {F.}~\bibnamefont {{Verstraete}}},\ }\bibfield  {title} {\bibinfo {title}
  {{Matrix product density operators: Renormalization fixed points and boundary
  theories}},\ }\href {https://doi.org/10.1016/j.aop.2016.12.030} {\bibfield
  {journal} {\bibinfo  {journal} {Annals of Physics}\ }\textbf {\bibinfo
  {volume} {378}},\ \bibinfo {pages} {100} (\bibinfo {year}
  {2017})}\BibitemShut {NoStop}%
\bibitem [{\citenamefont {Kato}(2024)}]{kato2024exact}%
  \BibitemOpen
  \bibfield  {author} {\bibinfo {author} {\bibfnamefont {K.}~\bibnamefont
  {Kato}},\ }\bibfield  {title} {\bibinfo {title} {Exact renormalization group
  flow for matrix product density operators},\ }\href
  {https://arxiv.org/abs/2410.22696} {\bibfield  {journal} {\bibinfo  {journal}
  {arXiv:2410.22696}\ } (\bibinfo {year} {2024})}\BibitemShut {NoStop}%
\bibitem [{\citenamefont {Liu}\ \emph {et~al.}(2025)\citenamefont {Liu},
  \citenamefont {Ruiz-de Alarc{\'o}n}, \citenamefont {Styliaris}, \citenamefont
  {sun}, \citenamefont {P{\'e}rez-Garc{\'\i}a},\ and\ \citenamefont
  {Cirac}}]{liu2025parent}%
  \BibitemOpen
  \bibfield  {author} {\bibinfo {author} {\bibfnamefont {Y.}~\bibnamefont
  {Liu}}, \bibinfo {author} {\bibfnamefont {A.}~\bibnamefont {Ruiz-de
  Alarc{\'o}n}}, \bibinfo {author} {\bibfnamefont {G.}~\bibnamefont
  {Styliaris}}, \bibinfo {author} {\bibfnamefont {X.-Q.}\ \bibnamefont {sun}},
  \bibinfo {author} {\bibfnamefont {D.}~\bibnamefont {P{\'e}rez-Garc{\'\i}a}},\
  and\ \bibinfo {author} {\bibfnamefont {J.~I.}\ \bibnamefont {Cirac}},\
  }\bibfield  {title} {\bibinfo {title} {Parent lindbladians for matrix product
  density operators},\ }\href {https://arxiv.org/abs/2501.10552} {\bibfield
  {journal} {\bibinfo  {journal} {arXiv:2501.10552}\ } (\bibinfo {year}
  {2025})}\BibitemShut {NoStop}%
\bibitem [{\citenamefont {Lessa}\ \emph {et~al.}(2025)\citenamefont {Lessa},
  \citenamefont {Cheng},\ and\ \citenamefont {Wang}}]{lessa2025mixed}%
  \BibitemOpen
  \bibfield  {author} {\bibinfo {author} {\bibfnamefont {L.~A.}\ \bibnamefont
  {Lessa}}, \bibinfo {author} {\bibfnamefont {M.}~\bibnamefont {Cheng}},\ and\
  \bibinfo {author} {\bibfnamefont {C.}~\bibnamefont {Wang}},\ }\bibfield
  {title} {\bibinfo {title} {Mixed-state quantum anomaly and multipartite
  entanglement},\ }\href {https://doi.org/10.1103/PhysRevX.15.011069}
  {\bibfield  {journal} {\bibinfo  {journal} {Phys. Rev. X}\ }\textbf {\bibinfo
  {volume} {15}},\ \bibinfo {pages} {011069} (\bibinfo {year}
  {2025})}\BibitemShut {NoStop}%
\bibitem [{\citenamefont {Xu}\ and\ \citenamefont
  {Zhang}(2018)}]{xu2018tensor}%
  \BibitemOpen
  \bibfield  {author} {\bibinfo {author} {\bibfnamefont {W.-T.}\ \bibnamefont
  {Xu}}\ and\ \bibinfo {author} {\bibfnamefont {G.-M.}\ \bibnamefont {Zhang}},\
  }\bibfield  {title} {\bibinfo {title} {Tensor network state approach to
  quantum topological phase transitions and their criticalities of {$Z_2$}
  topologically ordered states},\ }\href
  {https://doi.org/10.1103/PhysRevB.98.165115} {\bibfield  {journal} {\bibinfo
  {journal} {Phys. Rev. B}\ }\textbf {\bibinfo {volume} {98}},\ \bibinfo
  {pages} {165115} (\bibinfo {year} {2018})}\BibitemShut {NoStop}%
\bibitem [{Note1()}]{Note1}%
  \BibitemOpen
  \bibinfo {note} {Here, $\rho _{\protect \text {bdy}}^{(2N)}=\protect \mathcal
  {E}^{\otimes N}(\rho _{CZX}^{(N)})$ where $\rho _{CZX}^{(N)}=\DOTSB \prod@
  \slimits@ _{i=1}^{N} CZ_{i,i+1} \DOTSB \prod@ \slimits@ _{i=1}^{N} X_i$ and
  $u^{\otimes N}$ is the unitary that relates $\rho _{CZX}$ and $\rho _{CZY}$.
  Strictly speaking, the mapping from $\rho _{CZX}$ to $\rho _{CZY}$ is exact
  only for $N=0 \protect \mod 4$. For $N=0\protect \mod 2$, we can use
  $(u\otimes u^\dagger )^{\otimes N/2}$}\BibitemShut {NoStop}%
\bibitem [{\citenamefont {{Coser}}\ and\ \citenamefont
  {{P{\'e}rez-Garc{\'\i}a}}(2019)}]{coser2019classification}%
  \BibitemOpen
  \bibfield  {author} {\bibinfo {author} {\bibfnamefont {A.}~\bibnamefont
  {{Coser}}}\ and\ \bibinfo {author} {\bibfnamefont {D.}~\bibnamefont
  {{P{\'e}rez-Garc{\'\i}a}}},\ }\bibfield  {title} {\bibinfo {title}
  {{Classification of phases for mixed states via fast dissipative
  evolution}},\ }\href {https://doi.org/10.22331/q-2019-08-12-174} {\bibfield
  {journal} {\bibinfo  {journal} {Quantum}\ }\textbf {\bibinfo {volume} {3}},\
  \bibinfo {pages} {174} (\bibinfo {year} {2019})}\BibitemShut {NoStop}%
\bibitem [{\citenamefont {de~Groot}\ \emph {et~al.}(2022)\citenamefont
  {de~Groot}, \citenamefont {Turzillo},\ and\ \citenamefont
  {Schuch}}]{de2022symmetry}%
  \BibitemOpen
  \bibfield  {author} {\bibinfo {author} {\bibfnamefont {C.}~\bibnamefont
  {de~Groot}}, \bibinfo {author} {\bibfnamefont {A.}~\bibnamefont {Turzillo}},\
  and\ \bibinfo {author} {\bibfnamefont {N.}~\bibnamefont {Schuch}},\
  }\bibfield  {title} {\bibinfo {title} {Symmetry protected topological order
  in open quantum systems},\ }\href {https://doi.org/10.22331/q-2022-11-10-856}
  {\bibfield  {journal} {\bibinfo  {journal} {Quantum}\ }\textbf {\bibinfo
  {volume} {6}},\ \bibinfo {pages} {856} (\bibinfo {year} {2022})}\BibitemShut
  {NoStop}%
\bibitem [{\citenamefont {{Ma}}\ and\ \citenamefont
  {{Wang}}(2023)}]{ma2023average}%
  \BibitemOpen
  \bibfield  {author} {\bibinfo {author} {\bibfnamefont {R.}~\bibnamefont
  {{Ma}}}\ and\ \bibinfo {author} {\bibfnamefont {C.}~\bibnamefont {{Wang}}},\
  }\bibfield  {title} {\bibinfo {title} {{Average Symmetry-Protected
  Topological Phases}},\ }\href {https://doi.org/10.1103/PhysRevX.13.031016}
  {\bibfield  {journal} {\bibinfo  {journal} {Physical Review X}\ }\textbf
  {\bibinfo {volume} {13}},\ \bibinfo {eid} {031016} (\bibinfo {year}
  {2023})}\BibitemShut {NoStop}%
\bibitem [{\citenamefont {Sang}\ \emph {et~al.}(2024)\citenamefont {Sang},
  \citenamefont {Zou},\ and\ \citenamefont {Hsieh}}]{sang2024mixed}%
  \BibitemOpen
  \bibfield  {author} {\bibinfo {author} {\bibfnamefont {S.}~\bibnamefont
  {Sang}}, \bibinfo {author} {\bibfnamefont {Y.}~\bibnamefont {Zou}},\ and\
  \bibinfo {author} {\bibfnamefont {T.~H.}\ \bibnamefont {Hsieh}},\ }\bibfield
  {title} {\bibinfo {title} {Mixed-state quantum phases: Renormalization and
  quantum error correction},\ }\href
  {https://doi.org/10.1103/PhysRevX.14.031044} {\bibfield  {journal} {\bibinfo
  {journal} {Physical Review X}\ }\textbf {\bibinfo {volume} {14}},\ \bibinfo
  {pages} {031044} (\bibinfo {year} {2024})}\BibitemShut {NoStop}%
\bibitem [{\citenamefont {Ellison}\ and\ \citenamefont
  {Cheng}(2025)}]{ellison2024towards}%
  \BibitemOpen
  \bibfield  {author} {\bibinfo {author} {\bibfnamefont {T.~D.}\ \bibnamefont
  {Ellison}}\ and\ \bibinfo {author} {\bibfnamefont {M.}~\bibnamefont
  {Cheng}},\ }\bibfield  {title} {\bibinfo {title} {Towards a classification of
  mixed-state topological orders in two dimensions},\ }\href
  {https://doi.org/10.1103/PRXQuantum.6.010315} {\bibfield  {journal} {\bibinfo
   {journal} {PRX Quantum}\ }\textbf {\bibinfo {volume} {6}},\ \bibinfo {pages}
  {010315} (\bibinfo {year} {2025})}\BibitemShut {NoStop}%
\bibitem [{\citenamefont {Sohal}\ and\ \citenamefont
  {Prem}(2025)}]{sohal2025noisy}%
  \BibitemOpen
  \bibfield  {author} {\bibinfo {author} {\bibfnamefont {R.}~\bibnamefont
  {Sohal}}\ and\ \bibinfo {author} {\bibfnamefont {A.}~\bibnamefont {Prem}},\
  }\bibfield  {title} {\bibinfo {title} {Noisy approach to intrinsically
  mixed-state topological order},\ }\href
  {https://doi.org/10.1103/PRXQuantum.6.010313} {\bibfield  {journal} {\bibinfo
   {journal} {PRX Quantum}\ }\textbf {\bibinfo {volume} {6}},\ \bibinfo {pages}
  {010313} (\bibinfo {year} {2025})}\BibitemShut {NoStop}%
\bibitem [{\citenamefont {Sun}(2025)}]{sun2025anomalous}%
  \BibitemOpen
  \bibfield  {author} {\bibinfo {author} {\bibfnamefont {X.-Q.}\ \bibnamefont
  {Sun}},\ }\bibfield  {title} {\bibinfo {title} {Anomalous matrix product
  operator symmetries and 1d mixed-state phases},\ }\href
  {https://arxiv.org/abs/2504.16985} {\bibfield  {journal} {\bibinfo  {journal}
  {arXiv:2504.16985}\ } (\bibinfo {year} {2025})}\BibitemShut {NoStop}%
\bibitem [{\citenamefont {{Ruiz-de-Alarc{\'o}n}}\ \emph
  {et~al.}(2024)\citenamefont {{Ruiz-de-Alarc{\'o}n}}, \citenamefont
  {Garre-Rubio}, \citenamefont {Moln{\'a}r},\ and\ \citenamefont
  {P{\'e}rez-Garc{\'i}a}}]{ruiz2024matrix}%
  \BibitemOpen
  \bibfield  {author} {\bibinfo {author} {\bibfnamefont {A.}~\bibnamefont
  {{Ruiz-de-Alarc{\'o}n}}}, \bibinfo {author} {\bibfnamefont {J.}~\bibnamefont
  {Garre-Rubio}}, \bibinfo {author} {\bibfnamefont {A.}~\bibnamefont
  {Moln{\'a}r}},\ and\ \bibinfo {author} {\bibfnamefont {D.}~\bibnamefont
  {P{\'e}rez-Garc{\'i}a}},\ }\bibfield  {title} {\bibinfo {title} {Matrix
  product operator algebras {II}: phases of matter for {1D} mixed states},\
  }\href {https://doi.org/10.1007/s11005-024-01778-z} {\bibfield  {journal}
  {\bibinfo  {journal} {Letters in Mathematical Physics}\ }\textbf {\bibinfo
  {volume} {114}},\ \bibinfo {pages} {43} (\bibinfo {year} {2024})}\BibitemShut
  {NoStop}%
\bibitem [{\citenamefont {Assem}\ \emph {et~al.}(2006)\citenamefont {Assem},
  \citenamefont {Simson},\ and\ \citenamefont
  {Skowro{\'n}ski}}]{assem2006elements}%
  \BibitemOpen
  \bibfield  {author} {\bibinfo {author} {\bibfnamefont {I.}~\bibnamefont
  {Assem}}, \bibinfo {author} {\bibfnamefont {D.}~\bibnamefont {Simson}},\ and\
  \bibinfo {author} {\bibfnamefont {A.}~\bibnamefont {Skowro{\'n}ski}},\
  }\href@noop {} {\emph {\bibinfo {title} {Elements of the representation
  theory of associative algebras: Techniques of representation theory}}}\
  (\bibinfo  {publisher} {Cambridge University Press},\ \bibinfo {year}
  {2006})\BibitemShut {NoStop}%
\bibitem [{\citenamefont {B\"{o}hm}\ and\ \citenamefont
  {Szlach\'{a}nyi}(1996)}]{bohm_coassociativec_1996}%
  \BibitemOpen
  \bibfield  {author} {\bibinfo {author} {\bibfnamefont {G.}~\bibnamefont
  {B\"{o}hm}}\ and\ \bibinfo {author} {\bibfnamefont {K.}~\bibnamefont
  {Szlach\'{a}nyi}},\ }\bibfield  {title} {\bibinfo {title} {A {Coassociative}
  {C}*-{Quantum} {Group} with {Non}-{Integral} {Dimensions}},\ }\href
  {https://doi.org/10.1007/BF01815526} {\bibfield  {journal} {\bibinfo
  {journal} {Letters in Mathematical Physics}\ }\textbf {\bibinfo {volume}
  {38}},\ \bibinfo {pages} {437} (\bibinfo {year} {1996})}\BibitemShut
  {NoStop}%
\bibitem [{\citenamefont {B\"{o}hm}\ \emph {et~al.}(1999)\citenamefont
  {B\"{o}hm}, \citenamefont {Nill},\ and\ \citenamefont
  {Szlach\'{a}nyi}}]{bohm_weak_1999}%
  \BibitemOpen
  \bibfield  {author} {\bibinfo {author} {\bibfnamefont {G.}~\bibnamefont
  {B\"{o}hm}}, \bibinfo {author} {\bibfnamefont {F.}~\bibnamefont {Nill}},\
  and\ \bibinfo {author} {\bibfnamefont {K.}~\bibnamefont {Szlach\'{a}nyi}},\
  }\bibfield  {title} {\bibinfo {title} {Weak {Hopf} {Algebras} {I}: {Integral}
  {Theory} and {C}*-structure},\ }\href
  {https://doi.org/10.1006/jabr.1999.7984} {\bibfield  {journal} {\bibinfo
  {journal} {Journal of Algebra}\ }\textbf {\bibinfo {volume} {221}},\ \bibinfo
  {pages} {385} (\bibinfo {year} {1999})}\BibitemShut {NoStop}%
\bibitem [{\citenamefont {B\"{o}hm}\ and\ \citenamefont
  {Szlach\'{a}nyi}(2000)}]{bohm_weak_2000}%
  \BibitemOpen
  \bibfield  {author} {\bibinfo {author} {\bibfnamefont {G.}~\bibnamefont
  {B\"{o}hm}}\ and\ \bibinfo {author} {\bibfnamefont {K.}~\bibnamefont
  {Szlach\'{a}nyi}},\ }\bibfield  {title} {\bibinfo {title} {Weak {Hopf}
  {Algebras} {II}: {Representation} theory, dimensions and the {Markov}
  trace},\ }\href {https://doi.org/10.1006/jabr.2000.8379} {\bibfield
  {journal} {\bibinfo  {journal} {Journal of Algebra}\ }\textbf {\bibinfo
  {volume} {233}},\ \bibinfo {pages} {156} (\bibinfo {year}
  {2000})}\BibitemShut {NoStop}%
\bibitem [{\citenamefont {Garre-Rubio}\ \emph {et~al.}(2023)\citenamefont
  {Garre-Rubio}, \citenamefont {Lootens},\ and\ \citenamefont
  {Molnár}}]{garrerubio2023classifying}%
  \BibitemOpen
  \bibfield  {author} {\bibinfo {author} {\bibfnamefont {J.}~\bibnamefont
  {Garre-Rubio}}, \bibinfo {author} {\bibfnamefont {L.}~\bibnamefont
  {Lootens}},\ and\ \bibinfo {author} {\bibfnamefont {A.}~\bibnamefont
  {Molnár}},\ }\bibfield  {title} {\bibinfo {title} {Classifying phases
  protected by matrix product operator symmetries using matrix product
  states},\ }\href {https://doi.org/10.22331/q-2023-02-21-927} {\bibfield
  {journal} {\bibinfo  {journal} {Quantum}\ }\textbf {\bibinfo {volume} {7}},\
  \bibinfo {pages} {927} (\bibinfo {year} {2023})}\BibitemShut {NoStop}%
\end{thebibliography}%
\end{document}